\documentclass[journal]{IEEEtran}

\IEEEoverridecommandlockouts  
																													
\usepackage{amsmath}
\usepackage{amssymb}

\usepackage{amsthm}
\usepackage{amsfonts}
\usepackage{mathbbol} 
\usepackage{graphicx}
\usepackage{subfigure}
\usepackage{cite}
\usepackage{array}
\usepackage{upgreek}
\usepackage[dvipsnames]{xcolor}

\usepackage{tikz} 
\usepackage{pgfplots}

\usepackage[linesnumbered,ruled,vlined]{algorithm2e} 

\newtheorem{defn}{Definition}
\newtheorem{thm}{Theorem}
\newtheorem{lem}{Lemma}
\newtheorem{prop}{Proposition}
\newtheorem{cor}{Corollary}

\newtheorem{rem}{Remark}
\newtheorem{exam}{Example}
\newtheorem{assump}{Assumption}
\newtheorem{prob}{Problem}

\newcommand{\V}{\mathbb{V}}
\newcommand{\R}{\mathbb{R}}
\newcommand{\EE}{\mathbb{E}}
\newcommand{\E}{\mathcal{E}}
\newcommand{\G}{\mathcal{G}}

\newcommand{\IM}{\mathrm{Im}}

\newcommand{\diag}{\mathrm{diag}}
\newcommand{\Refc}{\mathrm{ref}}

\usepackage[color=yellow, textsize=small, textwidth=\marginparwidth, draft]{todonotes}   

\begin{document}

\title{Model-Free Practical Cooperative Control for Diffusively Coupled Systems\if(0) via Passivity Theory and Network Optimization\fi}

\author{Miel Sharf, Anne Koch, Daniel Zelazo and Frank Allg{\"o}wer \thanks{M. Sharf is with the Division of Decision and Control Systems, KTH Royal Institute of Technology, Stockholm, Sweden. {\tt\small sharf@kth.se}. A. Koch and F. Allg\"ower are with the Institute for Systems Theory and Automatic Control, University of Stuttgart, Germany. {\tt\small \{anne.koch,frank.allgower\}@ist.uni-stuttgart.de}. D. Zelazo is with the Faculty of Aerospace Engineering, Israel Institute of Technology, Haifa, Israel {\tt\small dzelazo@technion.ac.il}. This work was supported by the German-Israeli Foundation for Scientific Research and Development.}}

\maketitle
\begin{abstract}
In this paper, we develop a data-based controller design framework for diffusively coupled systems with guaranteed convergence to an $\epsilon$-neighborhood of the desired formation. The controller is comprised of a fixed controller with an adjustable gain on each edge. Via passivity theory and network optimization we not only prove that there exists a gain attaining the desired formation control goal, but we present a data-based method to find an upper bound on this gain. Furthermore, by allowing for additional experiments, the conservatism of the upper bound can be reduced via iterative sampling schemes. 
The introduced scheme is based on the assumption of passive systems, which we relax by discussing different methods for estimating the systems' passivity shortage, as well as applying transformations passivizing them. 
Finally, we illustrate the developed model-free cooperative control scheme with a case study.
\end{abstract}
\vspace{-10pt}
\section{Introduction}\label{sec.Intro}
Multi-agent systems have received extensive attention in the past years, due to their appearance in many fields of engineering, exact sciences and social sciences. Examples include robotics\cite{Franchi2011}, traffic engineering \cite{Bando1995} and ecology \cite{Urban2001}. The state-of-the-art approach to model-based control for multi-agent systems offers rigorous stability analysis, performance guarantees and systematic insights into the considered problem. However, with the growing complexity of systems, the modeling process is reaching its limits. Obtaining a reliable mathematical model of the agents becomes a time-intensive and arduous task. 

At the same time, modern technology allows for gathering and storing more and more data from systems and processes, inciting an increasing interest in \emph{data-driven control}. \textcolor{black}{There are two main approaches for data-driven control. The first is model-based data-driven control, which uses data to identify a model from the problem, which is in turn used to solve the synthesis problem \cite{Recht2019,Dean2017}. In this case, the model estimation errors must be taken into account when solving the synthesis problem. The second is model-free control, which does not try to estimate a model for the system. The latter can be further bisected into approximate dynamic programming methods and direct policy search. The former evaluates a score for each state-action pair, and then obtains an optimal control policy using dynamic programming \cite{Gorges2019}, and the latter tries to find the optimal policy directly, e.g. by gradient descent or via a non-parametric description of the possible trajectories \cite{Coulson2018}. These methods have all been applied to multi-agent systems as well, with varying degrees of success \cite{Fattahi2019,Gorges2019,Jiang2018,Bianchini2017}.}

In this work, we develop a data-driven controller synthesis approach for multi-agent systems that comes with rigorous theoretical analysis and stability guarantees for the closed loop, with almost no assumptions on the agents and few measurements needed. 
Our approach is based on high-gain control as well as passivity. Some ideas on high-gain approaches to cooperative control can be found in \cite{Siljak1996} and references therein. In \cite{Zheng2018}, the authors provide a high-gain condition in the design of distributed $H_\infty$ controllers for platoons with undirected topologies, while there are also many approaches to (adaptively) tune the coupling weights, e.g. \cite{Yu2012}. Our approach provides an upper bound on a high-gain controller using passivity measures. Passivity properties of the components can provide sufficient abstractions of their detailed dynamical models for guaranteed control. Such passivity properties can be obtained from data as ongoing work shows (e.g., \cite{Montenbruck2016, Romer2017b, Romer2019}).  

Passivity is a well-known tool for controller synthesis \cite{Khalil2001}, which is useful, beyond convergence analysis, for its powerful properties such as compositionality. It was first introduced in the field of multi-agent systems in the seminal works of Arcak \cite{Arcak2007,Bai2011}, and was since explored in many variants in the context of multi-agent systems in many other works \cite{Pavlov2008,Hines2011,Burger2014,Sharf2017,Sharf2018a,Jain2018,Sharf2019a,Franchi2011}. We focus on the variant known as maximal equilibrium independent passivity (MEIP), presented in \cite{Burger2014}. \textcolor{black}{The notion of MEIP establishes a connection between multi-agent systems and network optimization, see \cite{Burger2014,Sharf2017,Sharf2018a,Jain2018,Sharf2019a}. Different synthesis problems have been solved using this network optimization framework assuming an exact model for each of the agents exists \cite{Sharf2017,Sharf2018a,Sharf2019f}. More precisely, one needs a perfect description of the steady-state input-output behavior of the agents. Thus, these methods cannot be applied in our case.}

As we said, this work generally studies the problem of controller synthesis for diffusively coupled systems.  The control objective is to converge to an $\epsilon$-neighborhood of a constant prescribed relative output vector.  That is, for some tolerance $\epsilon > 0$, we aim to design controllers so that the steady-state limit of the relative output is $\epsilon$-close to the prescribed values. The related problem of practical synchronization of multi-agent systems have been considered in \cite{Montenbruck2015}, in which the agents were assumed to be known up to some bounded additive disturbance. However, a nominal model was needed to get practical synchronization. It was also pursued in \cite{Kim2016}, where strong coupling was used to drive agents close to a common trajectory, but again, a model for the agents was needed.

\textcolor{black}{
Our contributions are as follows. We present a model-free data-driven method for solving the practical formation control problem. This is done by cascading a fixed controller and an adjustable gain on each edge. We show that this gain can be chosen to guarantee a solution to the practical formation control problem. We then provide schemes to compute this gain offline only from input-output data without any model of the agents. In fact, this gain can be computed from only three experiments (per agent), and it can become less conservative with further data samples. If iterative experiments can be performed, we also provide an approach for applying different gains over different edges to further reduce any conservatism.
We survey the different advantages for each of the methods and discuss their applicability in terms of the number of required measurements, or trade-offs in terms of energy. We also provide simulations presenting the effectiveness of the presented model-free control methods.} \textcolor{black}{To the best of the authors' knowledge, no prior works consider data-driven control of multi-agent systems using passivity. Furthermore, this is the first application of the network optimization framework of \cite{Burger2014,Sharf2018a} where the agents do not have an exact model.}

\paragraph*{Notations}
We employ use notions from algebraic graph theory \cite{Godsil2001}. An undirected graph $\mathcal{G}=(\mathbb{V},\mathbb{E})$ consists of finite sets of vertices $\mathbb{V}$ and edges $\mathbb{E} \subset \mathbb{V} \times \mathbb{V}$.  We denote the edge having ends $i$ and $j$ in $\mathbb{V}$ by $k=\{i,j\} \in \mathbb{E}$. For each edge $k$, we pick an arbitrary orientation and denote $k=(i,j)$. The incidence matrix $\mathcal{E}\in\mathbb{R}^{|\mathbb{E}|\times|\mathbb{V}|}$ of $\mathcal{G}$ is defined such that for an edge $k=(i,j)\in \mathbb{E}$, we have $[\mathcal{E}]_{ik} =+1$, $[\mathcal{E}]_{jk} =-1$, and $[\mathcal{E}]_{\ell k} =0$ for $\ell \neq i,j$. $\mathrm{diam}(\G)$ denotes the diameter of $\G$.

We also use notions from linear algebra. For every vector $v\in\mathbb{R}^n$, $\diag(v)$ denotes the $n\times n$ diagonal matrix with the entries of $v$ on its diagonal. The image of any linear map $T$ between vector spaces will be denoted by $\IM(T)$. Also, for two sets $A,B\subseteq \mathbb{R}^d$, we define $A+B=\{a+b:\ a\in A,\ b\in B\}$. Furthermore, $\|\cdot\|$ is the Euclidean norm.

Lastly, if $\Sigma$ is a dynamical system, and $M$ is a linear map of appropriate dimension, we can consider the cascaded system of $\Sigma$ and $M$. The cascade of $\Sigma$ after $M$ is denoted $\Sigma M$, and the cascade of $\Sigma$ before $M$ is denoted $M\Sigma$.
\vspace{-5pt}
\section{Background: Network Optimization and\\ Passivity in Cooperative Control}\label{sec.Background}
Our goal in this subsection is to describe the diffusive coupling networks studied in \cite{Burger2014}, and to present the passivity-based network optimization framework achieved for multi-agent systems. See also \cite{Sharf2017,Sharf2018a}.
\vspace{-10pt}
\subsection{Diffusively Coupled Systems and Steady-State Relations}
Diffusively coupled networks are composed of agents $\{\Sigma_i\}_{i\in \V}$ interacting over a graph $\G=(\V,\EE)$ using edge controllers $\{\Pi_e\}_{e\in \EE}$. Each vertex $i\in \V$ represents an agent and each edge $e\in\EE$ represents a controller. We model them as SISO dynamical systems:
\begin{align}
\Sigma_i:\ \begin{cases} \dot{x}_i = f_i(x_i,u_i) \\ y_i = h_i(x_i,u_i) \end{cases},
\Pi_e:\ \begin{cases} \dot{\eta}_e = \phi_e(\eta_e,\zeta_e) \\ \mu_e = \psi_e(\eta_e,\zeta_e),\end{cases}
\end{align} 
where the agents' state, input and output are $x_i,u_i,y_i$ respectively, and the controllers' state, input and output are $\eta_e,\zeta_e,\mu_e$ respectively. To understand the coupling of these systems, we consider the stacked inputs and outputs of the agents and controllers as $y = [y_1,...,y_{|\V|}]^T$, and similarly for $u,\zeta,\mu$. The system is connected via the relations $\zeta = \E^T y$ and $u = -\E\mu$, where $\E$ is the incidence matrix of the graph $\G$. In other words, if we stack all agents to a dynamical system $\Sigma$, and stack all controllers to a dynamical system $\Pi$, the closed-loop is the feedback connection of $\Sigma$ and $\E\Pi\E^T$. See Fig. \ref{fig.ClosedLoopNoGain} for an illustration of the network, which we will denote by $(\G,\Sigma,\Pi)$. 

\begin{figure} [!t] 
    \centering
    \includegraphics[width=0.45\textwidth]{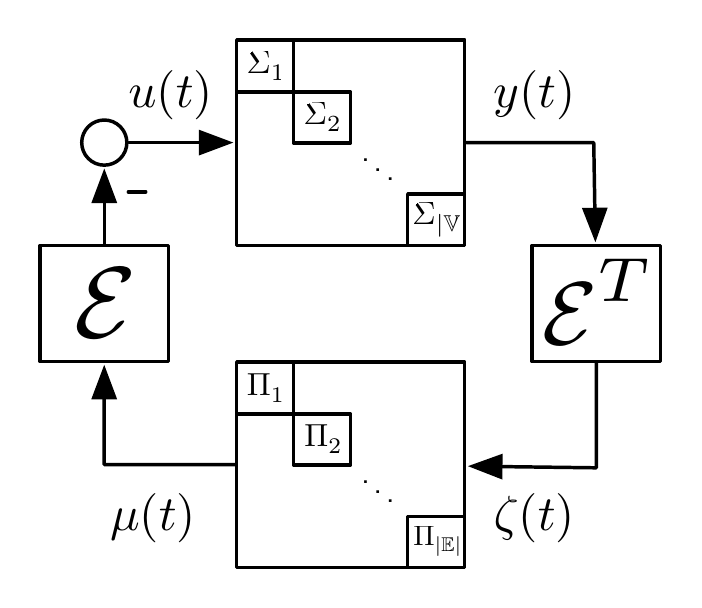}
    \caption{Block-diagram of the diffusively-coupled network $(\Sigma, \Pi, \mathcal{G})$.}
    \label{fig.ClosedLoopNoGain}
	\vspace{-12pt}
\end{figure}

We will be interested in steady-states of the closed-loop system. 
It's clear that if the stacked vectors $(\mathrm{u},\mathrm{y},\upzeta,\upmu)$ are a steady-state for $(\G,\Sigma,\Pi)$, then for every vertex $i\in \V$, $(\mathrm u_i, \mathrm y_i)$ is a steady-state input-output pair for the system $\Sigma_i$, and for every edge $e\in \EE$, $(\upzeta_e,\upmu_e)$ is a steady-state input-output pair for the system $\Pi_e$. This motivates the exploration of \emph{steady-state input-output relations}, first defined in \cite{Burger2014}.
\begin{defn}
The \emph{steady-state relation} of a system is a set containing all the steady-state input-output pairs of the system.
\end{defn}
We will denote the steady-states relations of $\Sigma_i,\Pi_e,\Sigma$, and $\Pi$ as $k_i,\gamma_e,k$, and $\gamma$, accordingly.
\begin{rem}
We will sometimes abuse the notation and consider this relation as a set-valued map. Indeed, for any input $\mathrm u$ we can define the set $k(\mathrm u)$ by
$
k(\mathrm u) = \{\mathrm y:\ (\mathrm{u,y})\in k\},
$
and similarly for $k_i,\gamma_e$, and $\gamma$. We also consider the inverse relation $k^{-1}$ as the set-valued map assigning to a steady-state output $\mathrm y$ the set
$
k^{-1}(\mathrm y) = \{\mathrm{u}:\ \mathrm y \in k(\mathrm u)\},
$
i.e., the set of all steady-state inputs corresponding to the steady-state output $\mathrm y$. We define this similarly for $k_i,\gamma_e$, and $\gamma$.
\end{rem} 
Thus, $(\mathrm{u,y},\upzeta,\upmu)$ is a steady-state of $(\G,\Sigma,\Pi)$ if and only if $\mathrm y \in k(\mathrm u)$, $\upmu \in \gamma(\upzeta)$, $\upzeta = \E^T \mathrm y$ and $\mathrm u = -\E\upmu$. Equivalently, $\mathrm y$ is a steady-state output of $(\G,\Sigma,\Pi)$ if and only if the zero vector $0$ lies in the set $k^{-1}(\mathrm y)+\E\gamma(\E^T\mathrm y)$ \cite{Burger2014,Sharf2017}.
\vspace{-10pt}
\subsection{Maximum Equilibrium-Independent Passivity and the Network Optimization Framework}
The main tool allowing us to connect multi-agent systems to the network optimization world is monotone relations.
\begin{defn}
A steady-state relation is \emph{monotone} if for any two points $(\mathrm u_1,\mathrm y_1)$ and $(\mathrm u_2,\mathrm y_2)$ in the relation, $u_1 < u_2$ implies $y_1 \le y_2$. We say that a monotone relation is \emph{maximally monotone} if it is not contained in a larger monotone relation.
\end{defn}
In order to connect this definition to the system-theoretic world, we define the following variant of passivity:
\begin{defn}[\cite{Burger2014}]
A SISO system is said to be (output-strictly) \emph{maximum equilibrium-independent passive} (MEIP) if the following two conditions hold:
\begin{itemize}
\item[i)] The system is (output-strictly) passive with respect to any steady-state input-output pair it has, and
\item[ii)] it's steady-state relation is maximally monotone.
\end{itemize}
\end{defn}

One important property of maximally monotone relations is that they are subgradients of convex functions \cite{Rockafeller1997}.
In this direction, we assume that the agents and controllers of the diffusively-coupled network $(\G,\Sigma,\Pi)$ are MEIP. Let $K_i$, and $\Gamma_e$ be the corresponding convex integral functions for the steady-state relations $k_i$,$\gamma_e$. In other words, $\partial K_i = k_i$ and $\partial \Gamma_e = \gamma_e$, where $\partial$ denotes the subdifferential of convex functions \cite{Rockafeller1997}. We shall denote $K= \sum_{i\in \V} K_i$ and $\Gamma = \sum_{e \in \mathbb{E}} \Gamma_e$, so that $\partial K = k$ and $\partial \Gamma = \gamma$. The dual functions of $K_i,\Gamma_e,K,\Gamma$ are defined using the Legendre transform, $K^\star(\mathrm y) = \sup_{\mathrm u}\{\mathrm u^T\mathrm y - K(\mathrm u)\} = -\inf_{\mathrm u}\{K(\mathrm u) - \mathrm u^T\mathrm y\}$, and similarly for $K_i^\star,\Gamma_e^\star$ and $\Gamma^\star$. We note that $\partial K^\star = k^{-1}$, $\partial \Gamma^\star = \gamma^{-1}$, $\partial K_i^\star = k_i^{-1}$ and $\partial \Gamma_e^\star = \gamma_e^{-1}$ \cite{Rockafeller1997}.

We now resume our interest in steady-states for the diffusively coupled network $(\G,\Sigma,\Pi)$. We recall that $\mathrm y$ was the steady-state output of the diffusively coupled network if and only if $0 \in k^{-1}(\mathrm y)+\E\gamma(\E^T\mathrm y)$. Restating this result in the language of convex functions gives the following theorem.
\begin{thm}[\cite{Burger2014}]\label{thm.BasicConvergenceMEIP}
Consider the diffusively coupled network $(\G,\Sigma,\Pi)$. Assume all agents $\Sigma_i$ are MEIP, and all controllers $\Pi_e$ are output-strictly MEIP 
(or vice versa). Let $k_i,\gamma_e,k$ and $\gamma$ be the steady-state relations of $\Sigma_i,\Pi_e,\Sigma$ and $\Pi$ accordingly, and let $K_i,\Gamma_e,K,\Gamma$ be the corresponding convex integral functions. Then the closed-loop system converges to a steady-state $(\mathrm{u,y},\upzeta,\upmu)$, such that $(\mathrm{y},\upzeta)$ and $(\mathrm{u},\upmu)$ are dual optimal solutions to the following convex optimization problems:
\begin{center}
\begin{tabular}{ c||c }
 \textbf{Optimal Potential Problem} & \textbf{Optimal Flow Problem}  \\ \hline
 $ \begin{array}{cl} \underset{y,\zeta}{\min} &K^\star(y) + \Gamma(\zeta)\\
\mathrm{s.t.}&\mathcal{E}^Ty = \zeta 
\end{array} $&  $ \begin{array}{cl}\underset{u,\mu}{\min}& K(u) + \Gamma^\star(\mu) \\
\mathrm{s.t.} &\mu = -\mathcal{E}u.
\end{array} $ 
\end{tabular}
\end{center}
\end{thm}
The two network optimization problems above will be denoted often by (OPP) and (OFP) respectively. These problems are fundamental in the field of network optimization, dealing with optimization problems defined on graphs \cite{Rockafeller1997}. The names ``optimal potential problem" and ``optimal flow problem" are inspired from standard nomenclature in this field.
\vspace{-6pt}
\section{Problem Formulation} \label{sec.Problem_Formulation}
\vspace{-2pt}
We focus on relative-output based formation control. In this problem, the agents know the relative output $\upzeta_e = y_i-y_j$ with respect to their neighbors, and the control goal is to converge to a steady-state with prescribed relative outputs $\upzeta_e = \mathrm y_i - \mathrm y_j$. Examples include the consensus problem, in which all outputs must agree, as well as relative-position based formation control of robots, in which the robots are required to organize themselves in a desired spatial structure \cite{Oh2015}.

More specifically, we are given a graph $\G$ and agents $\Sigma$, and our goal is to design controllers $\Pi$ so that the signal $\zeta(t)$ of the diffusively coupled network $(\G,\Sigma,\Pi)$ will converge to a desired, given steady-state vector $\upzeta^\star$. One evident solution is to apply a (shifted) integrator as a controller. However, this solution will not always work even when the agents are MEIP.

\begin{exam}
Consider agents $\Sigma_i$ with integrator dynamics, together with (shifted) integrator controllers $\Pi_e$, where we desire consensus (i.e., $\upzeta^\star = 0$) over a connected graph $\G$,
\begin{align*}
\Sigma_i: \begin{cases}\dot{x}_i = u_i \\ y_i = x_i\end{cases}, \hspace{1cm}
\Pi_e: \begin{cases}\dot{\eta}_e = \zeta_e \\ \mu_e = \eta_e\end{cases}.
\end{align*}
The trajectories of the diffusively-coupled system can be understood by noting that the closed-loop system yields the second-order dynamics $\ddot{x} = -\E\E^T x$. Decomposing $x$ using a basis of eigenvectors of the graph Laplacian $\E\E^T$, which is a positive semi-definite matrix, we see that the trajectory of $x(t)$ oscillates around the consensus manifold $\{\mathrm{x}:\ \exists \lambda \in \mathbb{R}\ \mathrm{x} = \lambda\mathbb{1}_n\}$. Specifically, $x(t) - \frac{1}{n}\mathbb{1}_n^Tx(t)\mathbb{1}_n = \sum_{i=2}^n c_i\cos(\sqrt{\lambda_i}t+\varphi_i)v_i$, where $\lambda_2,\ldots,\lambda_n > 0$ are the non-trivial eigenvalues of the graph Laplacian, $v_2,\ldots,v_n$ are corresponding unit-length eigenvectors, and $c_i,\varphi_i$ are constants depending on the initial conditions $x(0),\eta(0)$. Thus $x(t)=y(t)$ does not converge anywhere, let alone to consensus. Moreover, $\zeta(t)$ does not converge as $t\to \infty$, as $\E\zeta(t) = \E\E^T y(t) = \sum_{i=2}^n \lambda_i c_i \cos(\sqrt{\lambda_i}t + \varphi_i)v_i$. Thus the integrator controller does not solve the formation control problem in this case.
\end{exam}

Even if the integrator would solve this problem in general, we would like more freedom in choosing the controller. In practice, one might want to design the controller to satisfy extra requirements (like $\mathcal{H}_2$- or $\mathcal{H}_\infty$-norm minimization, or making sure that certain restrictions on the behavior of the system are not broken). We do not try and satisfy these more complex requirements, but instead show that a large class of controllers can be used to solve the practical formation control problem. In turn, this allows one to choose from a wide range of controllers, and try and satisfy additional desired properties. \cite{Sharf2017} offers an algorithm solving the problem, assuming the agents are MEIP and a perfect model of them is known. This algorithm allows a lot of freedom in the choice of controllers. However, in practice we oftentimes have no exact model of the agents, or any closed-form model.
To formalize the goals we aim at, we define the notion of \emph{practical} formation control.
\begin{prob}\label{def.pracformcont}
Given a graph $\G$, agents $\Sigma$, a desired formation $\upzeta^\star\in\IM(\E^T)$, and an error margin $\varepsilon$, find a controller $\Pi$ so that the relative output vector $\zeta(t)$ of the network $(\G,\Sigma,\Pi)$ converges to some $\upzeta_0$ such that $\|\zeta^\star-\upzeta_0\|\le \varepsilon$.
\end{prob}
By choosing suitable error margins $\varepsilon$, practical formation control (compared to formation control) comprises no restriction or real drawback in any application case. Therefore, solving the practical formation control problem constitutes an interesting problem especially for unknown dynamics of the agents. Thus, we strive to develop an algorithm solving this practical formation control problem without a model of the agents while still providing rigorous guarantees.

The underlying idea of our approach is amplifying the controller output.
Consider the scenario depicted in Fig. \ref{fig.ClosedLoopGain}, where the graph $\G$, the agents $\Sigma$ and the nominal controller $\Pi$ are fixed, and the gain matrix $A$ is a diagonal matrix $A=\diag(\{a_e\}_{e\in\EE}) $ with positive entries. 
We will show in the following that when the gains $a_e$ become large enough, the controller dynamics $\Pi$ become much more emphasized than the agent dynamics $\Sigma$. By correctly choosing the nominal controller $\Pi$ according to $\upzeta^\star$, we can hence achieve arbitrarily close formations to $\upzeta^\star$, as the effect of the agents on the closed-loop dynamics will be dampened. We denote the diffusively-coupled system in Fig.~\ref{fig.ClosedLoopGain} as the 4-tuple $(\G,\Sigma,\Pi,A)$, or as $(\G,\Sigma,\Pi,a)$ where $a$ is the vector of diagonal entries of $A$. In case $A$ has uniform gains, i.e.,  $A=\alpha I$, we denote the system as $(\G,\Sigma,\Pi,\alpha\mathbb{1}_n)$ . We make an assumption in order to apply the network optimization framework of Theorem \ref{thm.BasicConvergenceMEIP}:
\begin{assump}\label{Assump}
The agents $\{\Sigma_i\}_{i\in\V}$ are all MEIP, and the chosen controllers $\{\Pi_e\}_{e\in\EE}$ are all output-strictly MEIP.
\end{assump}

\begin{figure} [!t] 
    \centering
    \includegraphics[width=0.48\textwidth]{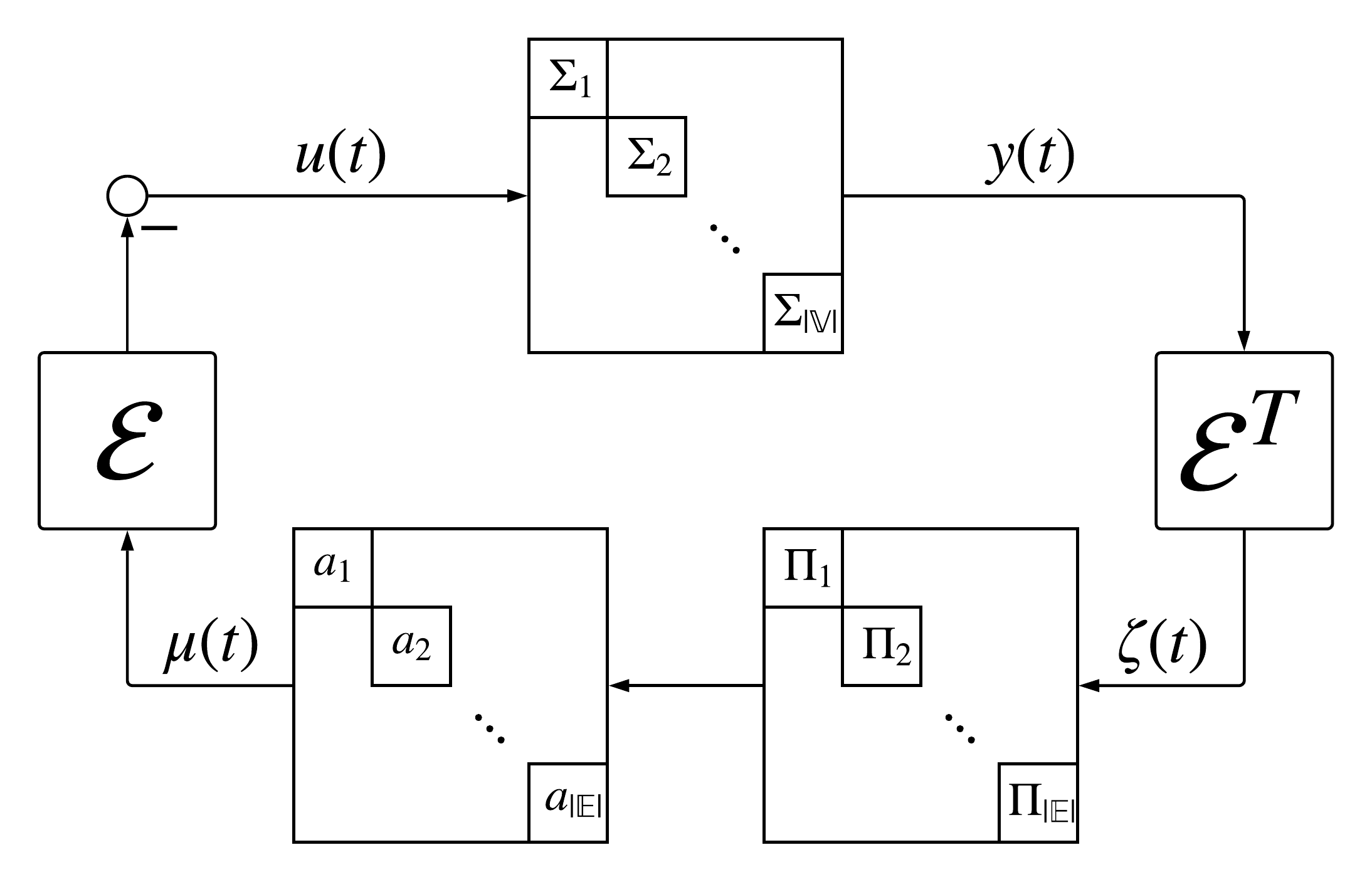}
    \caption{Block-diagram of the diffusively-coupled network $(\Sigma, \Pi, \mathcal{G},A)$.}
    \label{fig.ClosedLoopGain}
	\vspace{-15pt}
\end{figure}

Before expanding on the suggested controller design, we discuss Assumption \ref{Assump}. In practice, we might not know if an agent is MEIP. Hence, we discuss how to either verify MEIP for the agents, or otherwise determine their shortage of passivity. We also discuss how to passivize the agents in the latter case.
First, in some occasions, \textcolor{black}{we might not know a model for the agents, but some known general structure properties.} For example, one might know that an agent can be modeled by a gradient system, or a Euler-Lagrange system, but the exact model is unknown due to uncertainty on the characterizing parameters. In that case, we can use analytical results to verify MEIP. To exemplify this idea, we show how a very rough model can be used to prove that a system is MEIP. 
\begin{prop} \label{prop.MEIPFromObscureModel}
Consider a control-affine SISO system:
\begin{align}\label{controlaffine}
\dot{x} = -f(x) + g(x)u;\ y = h(x).
\end{align}
Assume that $g$ is positive, that $f/g,h:\mathbb{R}\to\mathbb{R}$ are continuous ascending functions, and that either $\lim_{|x|\to\infty} |f(x)/g(x)| = \infty$ or $\lim_{|x|\to\infty} |h(x)| = \infty$. Then \eqref{controlaffine} is MEIP.
\end{prop}
The proof is available in the appendix. See also \cite{Sharf2018a} for a treatment on gradient systems with oscillatory terms. Similarly, one can use a \textcolor{black}{highly uncertain} model to give an estimate about equilibrium-independent passivity indices.

Another approach for verifying Assumption~\ref{Assump} is learning input-output passivity properties from trajectories. For LTI systems, the shortage of passivity can be asymptotically revealed by iteratively probing the agents and measuring the output signal \cite{Romer2017b}. In \cite{Romer2019}, the authors showed that even one input-output trajectory (with persistently exciting input) is sufficient to find the shortage of passivity of an LTI system. For nonlinear agents, one can apply approaches presented in \cite{Montenbruck2016, Romer2017a}, under an assumption on Lipschitz continuity of the steady-state relation. However, for general non-linear systems, this is still a work in progress. We note that for LTI systems, output-strict passivity directly implies output strict MEIP \cite{Burger2014}. 

Using either approach, we can either find that an agent is MEIP, or that it has some shortage of passivity, and we need to render the agent passive in order to apply the model-free control approaches presented in this paper. We can use passivizing transformations in order to get a passive augmented agent. For example, if the agent has output-shortage of passivity $s_i > 0$, we can apply a controller $C_i: y_i \mapsto \nu_i y_i$ to the agent as in \cite{Jain2018}, with $\nu_i > s_i$, as shown in Fig.~\ref{fig:cl}. It can be shown that the augmented agent is output-strictly MEIP in this case. More generally, one could deal with more complex shortages of passivity, namely simultaneous input- and output-shortage of passivity, using more complex transformations \cite{Sharf2019a}.

\begin{figure}
\begin{center}
\begin{tikzpicture}
\draw[thick, fill=white!10] (0,0) rectangle (1.5,1);
\draw[thick] (0,-1.5) rectangle (1.5,-.5);
\draw[thick,->] (1.5,0.5) -- (3.5,0.5);
\node[anchor=south east] at (2.75,0.5) {\small $y_{i}$};
\node[anchor=south east] at (-1.5,0.5) {\small $u_{i}$};
\node[anchor=south east] at (-.25,0.5) {};
\node[anchor=south east] at (-.25,-1) {};

\draw[-,thick] (0,-1) -- (-1,-1);

\draw[thick] (-1,0.5) circle(0.1);
\node[anchor=north east] at (-1,0.5) {\small $-$};

\draw[thick,-] (2.5,0.5)--(2.5,-1);
\draw[thick,->] (-1,-1) -- (-1,0.4);
\draw[thick,->] (-0.9,0.5) -- (0,0.5);
\draw[thick,->] (-2,0.5) -- (-1.1,0.5);
\draw[thick,->] (2.5,-1) -- (1.5,-1);

\node at (0.75,0.5) {\small $\Sigma_{i}$};
\node at (0.75,-1) {\small $C_{i}$};
\end{tikzpicture}
\end{center}
\caption{Passivation of a passivity-short agent using feedback.}
\label{fig:cl}
\vspace{-15pt}
\end{figure}
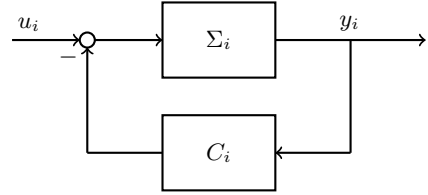

With this discussion and relaxation of Assumption~\ref{Assump}, we return to our solution of the practical formation control problem. We considered closed-loop systems of the form $(\G,\Sigma,\Pi,a)$, where $a$ is a vector of edge gains. From here, the paper diverges into two sections. The next section deals with theory and analysis for uniform edge gains. The following section deals with the case of heterogeneous edge gains.
\vspace{-5pt}
\section{Practical Formation Control\\ with Uniform Gains} \label{sec.SingleGain}
The chapter is split into two parts. The first part deals with the theory, and the second part deals with the corresponding implementation of practical formation control synthesis using uniform gains on the edges.
\vspace{-5pt}
\subsection{Theory}
We wish to understand the effect of amplification on the steady-state of the closed-loop system. For the remainder of the section, we fix a graph $\G$, agents $\Sigma$ and controllers $\Pi$ such that Assumption \ref{Assump} holds. We consider the diffusively coupled system $(\G,\Sigma,\Pi,\alpha \mathbb{1}_n)$ in Fig.~\ref{fig.ClosedLoopGain}, where the gains over all edges are identical and equal to $\alpha > 0$, and wish to understand the effect of $\alpha$. We let $K$ and $\Gamma$ denote the sum of the integral functions of the agents and of the controllers, respectively. We first study the steady-states of this diffusively coupled system.

\begin{lem} \label{lem.coupling_OPP}
Under the assumptions above, the closed-loop converges to a steady-state, and the steady-states {$\mathrm{y},\upzeta$} of the closed-loop system are minimizers of the following problem:
\begin{align} \label{eq.OPP_alpha}
\min_{y,\zeta}&\quad K^\star(y)+\alpha\Gamma(\zeta) \\
s.t.&\quad \E^Ty = \zeta. \nonumber
\end{align}
\end{lem}

\begin{proof}
We define a new stacked controller, $\bar{\Pi}=\alpha\Pi$, by cascading the previous controller $\Pi$ with the gain $\alpha$. {The resulting controller $\bar{\Pi}$ is again output-strictly MEIP}, and we let $\bar\gamma,\bar\Gamma$ denote the corresponding steady-state input-output relation and integral function. Theorem \ref{thm.BasicConvergenceMEIP} implies that the closed-loop system (with $\bar\Pi$) converges to minimizers of (OPP) for the system $(\G,\Sigma,\bar\Pi)$. Hence, we have $\bar\gamma(\zeta)=\alpha\gamma(\zeta)$ for any $\zeta\in\mathbb{R}^{|\EE|}$. Therefore, (OPP) for the system $(\G,\Sigma,\bar\Pi)$ reads:
\begin{align*}
\min_{y,\zeta}&\quad K^\star(y)+\alpha\Gamma(\zeta) \\
s.t.&\quad \E^Ty = \zeta,
\end{align*}
as $\bar\Gamma = \alpha\Gamma$.
\end{proof}

Our goal is to show that when $\alpha \gg 1$, the relative output vector $\zeta$ of the diffusively coupled system $(\G,\Sigma,\Pi,\alpha \mathbb{1}_n)$ globally asymptotically converges to an $\varepsilon=\varepsilon(\alpha)$-ball around the minimizer of $\Gamma$, and $\lim_{\alpha \to \infty} \varepsilon(\alpha) = 0$. Thus, if we design the controllers so that $\Gamma$ is minimized at $\zeta^\star$, then $\alpha \gg 1$ provides a solution to the $\varepsilon$-practical formation control problem. Indeed, we can prove the following theorem:

\begin{thm} \label{thm.practical_formation_control_coupling}
Consider the closed-loop system $(\G,\Sigma,\Pi,\alpha \mathbb{1}_n)$, where the agents are {MEIP} and the controllers are {output-strictly MEIP}. Assume $\Gamma$ has a unique minimizer in $\IM(\E^T)$, denoted $\zeta_1$. For any $\varepsilon>0$, there exists some $\alpha_0>0$, such that for all $\alpha > \alpha_0$ and for all initial conditions, the closed-loop converges to a vector $\mathrm y$ satisfying $\|\E^T\mathrm y-\zeta_1\|<\varepsilon$. In particular, if $\zeta_1 = \zeta^\star$, we solve the practical control problem.
\end{thm}

In order to prove the theorem, we study (OPP) for the diffusively coupled system $(\G,\Sigma,\Pi,\alpha \mathbb{1}_n)$, as described in Lemma \ref{lem.coupling_OPP}. In order to do so, we need to prove a couple of lemmas. The first deals with lower bounds on the values of convex functions away from their minimizers.
\begin{lem} \label{lem.inverse_continuity_convex}
Let $U$ be a finite-dimensional vector space. Let $f:U \to\R$ be convex, and suppose $x_0 \in U$ is the unique minimum of $f$. Then for any $\delta>0$ there is some $M>f(x_0)$ such that for any point $x\in U$, if $f(x)<M$ then $\|x-x_0\| < \delta$.
\end{lem}

\begin{proof}
We assume without loss of generality that $f(x_0)=0$. Let $\mu$ be the minimum of $f$ on the set $\{x\in U:\ ||x-x_0||=\delta\}$, which is positive since $x_0$ is $f$'s unique minimum and the set $\{x\in U:\|x-x_0\|=\delta\}$ is compact. We know that, for any $y\in U$, the difference quotient 
$\frac{f(x_0+\lambda y)-f(x_0)}{\lambda}$
is an increasing function of $\lambda>0$ (see Theorem 23.1 of \cite{Rockafeller1997}). Manipulating this inequality shows that for any $x\in U$, $||x||\ge \delta$ implies $f(x) \ge \frac{||x||}{\delta}\mu$, and in particular $f(x)\ge \mu$ whenever $||x|| \ge \delta$. Thus, if $f(x) < \mu$ then we must have $||x-x_0||<\delta$, so we can choose $M=\mu$ and complete the proof.
\end{proof}

The second lemma deals with minimizers of perturbed versions of convex functions on graphs.

\begin{lem} \label{lem.ApproximateOptimization}
Fix a graph $\G=(\V,\EE)$ and let $\E$ be its incidence matrix. Let $K:\R^{|\V|}\to\R$ be a convex function, and let $\Gamma:\R^{|\EE|}\to\R$ be a convex function with a unique minimum $\zeta_1 $ when restricted to the set $\IM(\E^T)$. For any $\alpha>0$, consider the function $F_\alpha(\mathrm y)=K^\star(\mathrm y)+\alpha\Gamma(\E^T\mathrm y)$. Then for any $\varepsilon>0$, there exists some $\alpha_0>0$ such that if $\alpha>\alpha_0$ then all of $F_\alpha$'s minima, $\mathrm y$, satisfy $\|\E^T\mathrm y-\zeta_1\|<\varepsilon$.
\end{lem}

\begin{proof}
By subtracting constants from $K{^\star}$ and $\Gamma$, we may assume without loss of generality that $\min(K{^\star})=\min(\Gamma)=0$.
Choose some $\mathrm y_0\in\R^{|\V|}$ such that $\E^T\mathrm y_0=\zeta_1$ and let $m=K{^\star}(\mathrm y_0)$. Note that $F_\alpha(y_0)=m$, meaning that if $\mathrm y$ is any minimum of $F_\alpha$, it must satisfy $F_\alpha(\mathrm y)\le m$, and in particular $\Gamma(\E^T\mathrm y)\le\frac{m}{\alpha}$.
Now, from Lemma \ref{lem.inverse_continuity_convex} we know that there is some $M>0$ such that if $\Gamma(\E^T \mathrm y)<M$ then $\|\E^T\mathrm y-\zeta_1\|<\varepsilon$. If we choose $\alpha_0=\frac{m}{M}$, then whenever $\alpha>\alpha_0$ we have $\Gamma(\E^T\mathrm y)<M$, implying $\|\E^T\mathrm y-\zeta_1\|<\varepsilon$. 
\end{proof}

We now connect the pieces and prove Theorem \ref{thm.practical_formation_control_coupling}.

\begin{proof}
Lemma \ref{lem.coupling_OPP} implies that the closed-loop system always converges to a minimizer of \eqref{eq.OPP_alpha}.
Lemma \ref{lem.ApproximateOptimization} proves that there exists $\alpha_0>0$ such that if $\alpha>\alpha_0$ then all minimizers of (OPP) satisfy $\|\E^T\mathrm y-\zeta_1\|<\varepsilon$. This proves the theorem.
\end{proof}

\begin{rem} \label{rem.ConstructiveGain}
The parameters $\varepsilon$ and $\zeta^\star$ can be used to estimate the minimal required gain $\alpha_0$ by following the proofs of Lemma \ref{lem.inverse_continuity_convex} and Lemma \ref{lem.ApproximateOptimization}. Namely, $\alpha_0 \le \frac{m}{M}$ where $M$ is the minimum of $\Gamma$ on the set $\{\zeta\in\IM(\E^T):\ \|\zeta-\zeta^\star\| = \varepsilon\}$, and $m = K^\star(\mathrm y_0) - \min_y{K^\star(\mathrm y)}=K^\star(\mathrm y^\star)+K(0)$ where $\mathrm y^\star\in \R^{|\V|}$ is any vector satisfying $\E^T \mathrm y^\star = \zeta^\star$.
\end{rem}

\begin{cor} \label{cor.formation_control_limit_simple_integrators}
Let $(\G,\Sigma, \Pi, \alpha \mathbb{1}_n)$ satisfy Assumption \ref{Assump} and let $(\G,\Sigma_{\mathrm{Int}},\Pi)$ be a network comprised of integrator dynamics for each agent. Denote the relative outputs of each system as $\zeta(t)$ and $\zeta_{\mathrm{Int}}(t)$ respectively.  Then for any $\varepsilon > 0$, there exists an $\alpha_0>0$ such that if $\alpha \ge \alpha_0$, then the relative outputs $\zeta(t)$ and $\zeta_{\mathrm{Int}}(t)$ both converge to constant vectors $\upzeta^\star$ and $\upzeta_{\mathrm{Int}}^\star$ respectively, and satisfy $\|\upzeta^\star - \upzeta_{\mathrm{Int}}^\star\| \leq \varepsilon$.
\end{cor}

\begin{proof}
The agents $\Sigma_{\mathrm{Int}}$ are MEIP. Thus, by Theorem \ref{thm.BasicConvergenceMEIP}, we know that the diffusively-coupled system $(\G,\Sigma_{\mathrm{Int}},\Pi)$ converges to a steady-state, and its steady-state output is a minimizer of the associated (OPP) problem. Note that the input-output relation of $\Sigma_{\mathrm{Int}}$ is given via $k^{-1}(\mathrm y) = 0$, meaning the integral function $K^\star$ is the zero function. Thus the associated problem (OPP) is the unconstrained minimization of $\Gamma(\E^T \mathrm y)$, meaning that the system $(\G,\Sigma_{\mathrm{Int}},\Pi)$ converges, and its output converges to a minimizer of $\Gamma(\E^T \mathrm y)$, i.e., its relative output $\zeta(t)$ converges to the minimizer of $\Gamma$ on $\IM(\E^T)$. Applying Theorem \ref{thm.practical_formation_control_coupling} now completes the proof.
\end{proof}

\begin{rem}[Almost Data-free control]\label{rem.data_free}
Corollary \ref{cor.formation_control_limit_simple_integrators} can be thought of as a synthesis procedure. Indeed, we can solve the synthesis problem as if the agents were single integrators, and then amplify the controller output by a factor $\alpha$. The corollary shows that for any $\varepsilon > 0$, there is a threshold $\alpha_0 > 0$ such that if $\alpha > \alpha_0$, the closed-loop system converges to an $\varepsilon$-neighborhood of $\zeta^\star$. 
We note that we only know that $\alpha_0$ exists as long as the agents are MEIP. Computing an estimate on $\alpha_0$, however, requires one to conduct a few experiments. 

There are a few possible approaches to overcome this requirement. One can try an iterative scheme, in which the edge gains are updated between iterations. Gradient-descent and extremum-seeking approaches are discussed in the next section (see Algorithm \ref{alg.MultiGain2}), but both require to measure the system between iterations. 
Another approach is to update the edge gains on a much slower time-scale than the dynamics of the system. This results in a two time-scale dynamical system, where the gains $a_e$ of the system $(\G,\Sigma,\Pi, a)$ are updated slowly enough to allow the system to converge. Taking $a_e$ as uniform gains of size $\alpha$, and slowly increasing $\alpha$, assures that eventually, $\alpha > \alpha_0$, so the system will converge $\varepsilon$-close to $\zeta^\star$. The only data we need is whether or not the system has already converged to an $\varepsilon$-neighborhood of $\zeta^\star$\textcolor{black}{\footnote{\textcolor{black}{Such data can be obtained by different methods, e.g. checking the size of the derivatives, using physical intuition in some cases, or using passivity to determine convergence rates as in \cite{Sharf2019f}.}}} , to know whether $\alpha$ should be updated or not. This requires no data on the trajectories themselves, nor information on the specific steady-state limit, \textcolor{black}{but only knowing whether the control goal has been achieved, which is the coarsest form of data. This results in an \emph{almost data-free}} solution of the practical formation control problem, which is valid as long as the agents are MEIP. 
\end{rem}
\vspace{-15pt}
\subsection{Data-Driven Determination of Gains}
In the previous subsection, we introduced a formula for a uniform gain $\alpha$ described by the ratio of $m$ and $M$, that solves the practical formation problem, where $m$ and $M$ are as defined in Remark \ref{rem.ConstructiveGain}. The parameter $M$ depends on the integral function $\Gamma$ of the controllers, evaluated on well-defined points, namely $\{\zeta\in\IM(\E^T):\ \|\zeta-\upzeta^\star\| = \varepsilon\}$. Thus we can compute $M$ exactly with no prior knowledge on the agents. This is not the case for the parameter $m$, which depends on the integral function of the agents. Without knowledge of any model of the agents, we need to obtain an estimate of $m$ solely on the basis of input-output data from the agents.

From Remark \ref{rem.ConstructiveGain}, we know that $m=\sum_{i=1}^n (K_i^\star(\mathrm y_i^\star) + K_i(0))=\sum_{i=1}^n m_i$ for some $\mathrm y^\star \in \mathbb{R}^n$ such that $\E^T \mathrm y^\star = \upzeta^\star$. Without a model of the agents, $m$ cannot be computed directly, but we can find an upper bound on $m$ from measured input-output trajectories via the inverse relations $k_i^{-1}$, $i=1, \dots, n$.

\begin{prop} \label{prop.estimate_m}
Let $(\mathrm u_i^\star,\mathrm y_i^\star)$, $(\mathrm u_{i,1}$, $\mathrm y_{i,1})$, $(\mathrm u_{i,2},\mathrm y_{i,2})$, $\dots$, $(\mathrm u_{i,r},\mathrm y_{i,r})$ and $(0,\mathrm y_{i,0})$ be steady-state input-output pairs for agent $i$, for some $r \ge 0$. Then:
$$
m_i \le \mathrm u_{i,1}(\mathrm y_{i,1}-\mathrm y_{i,0}) + \cdots + \mathrm u_{i,r}(\mathrm y_{i,r}-\mathrm y_{i,r-1}) + \mathrm u_i^\star(\mathrm y^\star_i - \mathrm y_{i,r}).
$$
\end{prop}
\begin{proof}
We prove the claim by induction on the number of steady-state pairs, $r+2$. First, consider the case $r=0$ of two steady-state pairs.
Because $(0,\mathrm y_{i,0})$ is a steady-state pair, we know that $K_i(0) = -K_i^\star(\mathrm y_{i,0})$ by Fenchel duality. Similarly, $K_i(\mathrm u_i^\star) = \mathrm u_i^\star \mathrm y_i^\star-K_i^\star(\mathrm y_i^\star)$. 
Thus,
\begin{align*} 
m_i = K_i^\star(\mathrm y_i^\star) {+} K_i(0) &= K_i^\star(\mathrm y_i^\star) {-} K_i^\star(\mathrm y_{i,0}) \le \mathrm u_i^\star (\mathrm y^\star_i {-} \mathrm y_{i,0}),
\end{align*}
where we use the inequality $K^\star_i(b)-K^\star_i(c) \ge k^{-1}_i(c)(b-c) $ for $b=\mathrm y_{i,0}$ and $c = \mathrm y^\star_i$. 
Now, we move to the case $r \ge 1$. We write $m_i$ as $(K_i^\star(\mathrm y_i^\star) - K_i^\star(\mathrm y_{i,r})) + (K_i^\star(\mathrm y_{i,r}) - K_i(0))$. The first element can be shown to be bounded by $\mathrm u_i^\star(\mathrm y^\star_i-\mathrm y_{i,r})$ by the case $r=0$. The second element is bounded by $\mathrm u_{i,1}(\mathrm y_{i,1}-\mathrm y_{i,0}) + \cdots + \mathrm u_{i,r}(\mathrm y_{i,r}-\mathrm y_{i,r-1})$ by induction hypothesis, as we use a total of $r+1$ steady-state pairs. Thus, $m_i$ is no greater than the sum of the two bounds, $\mathrm u_{i,1}(\mathrm y_{i,1}-\mathrm y_{i,0}) + \cdots + \mathrm u_{i,r}(\mathrm y_{i,r}-\mathrm y_{i,r-1}) + \mathrm u_i^\star(\mathrm y^\star_i - \mathrm y_{i,r})$.
\end{proof}
\begin{rem} \label{rem.estimate_m_2pt}
If we only have two steady-state pairs, $(\mathrm u_i^\star,\mathrm y_i^\star)$ and $(0,\mathrm y_{i,0})$, the estimate on $m_i$ becomes $m_i \le \mathrm u_i^\star (\mathrm y_i^\star-\mathrm y_{i,0})$. Thus two steady-state pairs, corresponding to two measurements/experiments, are enough to yield a meaningful bound on $m_i$. We do note that more experiments yield better estimates of $m_i$, i.e., if $r\ge 1$ the estimate in Proposition \ref{prop.estimate_m} is better as long as $(\mathrm y_{i,0},\mathrm y_{i,1},...,\mathrm y_{i,r},\mathrm y_i^\star)$ is a monotone series.
\end{rem}

With can use Remark~\ref{rem.estimate_m_2pt} to compute an upper bound on $m$ from the two steady-state pairs $(\mathrm u_i^\star,\mathrm y_i^\star)$ and $(0,\mathrm y_{i,0})$ per agent.  Designing experiments to measure these quantities is possible, but can require additional information on the plant, e.g. output-strict passivity. Instead, we take another path and estimate $\mathrm y_{i,0}$ and $\mathrm u_i^\star$ instead of computing them directly. Indeed, we use the monotonicity of the steady-state input-output relation to bound $\mathrm u_i^\star$ and $\mathrm y_{i,0}$ from above and below.  The approach is described in Algorithm \ref{alg.MEIPExperimentAndEstimate}. It is important to note that the closed-loop experiments are done with a output-strictly MEIP controller, which assure that the closed-loop system indeed converges.

\begin{algorithm} 
\caption{Estimating $m_i$ for an MEIP Agent}
\label{alg.MEIPExperimentAndEstimate}
\SetAlgoLined  
\textcolor{black}{\small Run the system in Fig.~\ref{fig.experiments} with $\beta_i$ small and $\mathrm y_{\Refc} = \frac{1}{\beta_i}$\;
Wait for convergence, and measure the steady-state output $\mathrm y_{i,+}$ and the steady-state input $\mathrm u_{i,+}$\;
Run the system in Fig.~\ref{fig.experiments} with $\beta_i$ small and $\mathrm y_{\Refc} = -\frac{1}{\beta_i}$\;
Wait for convergence, and measure the steady-state output $\mathrm y_{i,-}$ and the steady-state input $\mathrm u_{i,-}$\;
\eIf{$\mathrm y_{i,+} < \mathrm y_i^\star$}{
Run the system in Fig.~\ref{fig.experiments} with $\beta_i = 1$ and $\mathrm y_{\Refc} \gg \mathrm y_i^\star$\;
Wait for convergence, and measure the steady-state input $\mathrm u_{i,2}$ and output $\mathrm y_{i,2}$\;
}
{ 
Run the system in Fig.~\ref{fig.experiments} with $\beta_i = 1$, $\mathrm y_{\Refc} \ll \mathrm y_i^\star$\;
Wait for convergence, and measure the steady-state input $\mathrm u_{i,2}$ and output $\mathrm y_{i,2}$\;
}
Sort $\{\mathrm u_{i,-},\mathrm u_{i,+},\mathrm u_{i,2}\}$. Denote the result by $\{U_1,U_2,U_3\}$\;  \label{partalg.Start} 
Sort $\{\mathrm y_{i,-},\mathrm y_{i,+},\mathrm y_{i,2}\}$. Denote the result by $\{Y_1,Y_2,Y_3\}$\;
\eIf{$U_2> 0$} { 
Define $\underline{\mathrm y_{i,0}} = Y_1$ and $\overline{\mathrm y_{i,0}} = Y_2$\;
}
{ 
	Define $\underline{\mathrm y_{i,0}} = Y_2$ and $\overline{\mathrm y_{i,0}} = Y_3$\;
}
\eIf{$Y_2 > \mathrm y_i^\star$}
{
Define $\underline{\mathrm u_i^\star} = U_1$ and $\overline{\mathrm u_i^\star} = U_2$\;
}
{ 
	Define $\underline{\mathrm u_i^\star} = U_2$ and $\overline{\mathrm u_i^\star} = U_3$\;
}
\textbf{return} $\mathbb{m}_i$ as the maximum over $\omega(\mathrm y_i^\star - \upsilon)$, where $\omega \in \{ \underline{\mathrm u_i^\star},\overline{\mathrm u_i^\star}\}$ and $\upsilon \in \{ \underline{\mathrm y_{i,0}},\overline{\mathrm y_{i,0}}\}$\;}
\end{algorithm}

We prove the following: 

\begin{prop}
The output $\mathbb{m}_i$ of Algorithm \ref{alg.MEIPExperimentAndEstimate} is an upper bound on $m_i$.
\end{prop}
\begin{proof}
First, we show that the closed-loop experiments conducted by the algorithm indeed converge. The plant $\Sigma_i$ is assumed to be passive with respect to any steady-state input-output pair it possesses. Moreover, the static controller $u = \beta_i(y-y_{\rm ref})$ is output-strictly passive with respect to any steady-state input-output pair it possesses. Thus it is enough to show that the closed-loop system has a steady-state, which will prove convergence as this is a feedback connection of a passive system with an output-strictly passive system.
Indeed, a steady-state input-output pair $(\mathrm{u_i,y_i})$ of the system must satisfy $\mathrm u_i \in k_i^{-1}(\mathrm y_i)$ and $\mathrm u_i = -\beta_i(\mathrm y_i - \mathrm y_{\rm ref})$, or $-\beta_i(\mathrm{y_i-y_{ref}}) \in k^{-1}(\mathrm y_i)$. This is equivalent to
\[
0\in k^{-1}(\mathrm y_i) + \beta_i (\mathrm y_i - \mathrm y_{\rm ref}) = \nabla\bigg(K^\star_i(\mathrm y_i) + \frac{\beta_i}{2}(\mathrm y_i - \mathrm y_{\rm ref})^2\bigg),
\]
so $\mathrm y_i$ exists and is equal to the minimizer of $K^\star_i(\mathrm y_i) + \frac{\beta_i}{2}(\mathrm y_i - \mathrm y_{\rm ref})^2$. This shows that the closed-loop experiments converge. it remains to show that it outputs an upper-bound on $m_i$.

Using Remark \ref{rem.estimate_m_2pt}, it is enough to show that $\mathrm y_{i,0} \in [\underline{\mathrm y_{i,0}},\overline{\mathrm y_{i,0}}]$ and $\mathrm u_i^\star \in [\underline{\mathrm u_i^\star},\overline{\mathrm u_i^\star}]$. To do so, we first claim that $U_1 \le \mathrm u_i^\star \le U_3$ and $Y_1 \le \mathrm y_{i,0} \le Y_3$. 
We first show that $Y_1 \le \mathrm y_{i,0}$, by showing that $\mathrm y_{i,-} \le \mathrm y_{i,0}$. Indeed, because $k_i$ is a monotone map, this is equivalent to saying that $\mathrm u_{i,-} \le 0$. By the structure of the second experiment, the steady-state input is close to $-1$, and in particular smaller than $0$. The inequality $\mathrm y_{i,0} \le \mathrm y_{i,+}$ is proved similarly. We note that because $\mathrm u_{i,-} \approx -1$ and $\mathrm u_{i,+} \approx 1$, we have $\mathrm u_{i,-} \le \mathrm u_{i,+}$ and thus $\mathrm y_{i,-} \le \mathrm y_{i,+}$. as $k_i$ is monotone.

Next, we prove that $U_1 \le \mathrm u_i^\star$. By monotonicity of $k_i$, this is equivalent to $Y_1 \le \mathrm y_i^\star$. Because $\mathrm y_{i,-} \le \mathrm y_{i,+}$, it is enough to show that either $\mathrm y_{i,-} \le \mathrm y_i^\star$ or $\mathrm y_{i,2} \le \mathrm y_i^\star$. If the first case is true, then the proof is complete. Otherwise, $\mathrm y_{i,-} > \mathrm y_i^\star$, so the algorithm finds $\mathrm y_{i,2}$ by running the closed-loop system in Fig.~\ref{fig.experiments} with $\beta_i = 1$ and $\mathrm y_\mathrm{ref} \ll y_i^\star$. The increased coupling strength implies that the steady-state output $\mathrm y_{i,2}$ should be close to $\mathrm y_{\mathrm{ref}}$, which is much smaller than $\mathrm y_i^\star$. Thus $\mathrm y_{i,2} < \mathrm y_i^\star$, which shows that $Y_1 \le \mathrm y_1^\star$, or equivalently $U_1 \le \mathrm u_1^\star$. The proof that $\mathrm u_1^\star \le U_3$ is similar. This completes the proof.
\end{proof}
\begin{rem}
Algorithm \ref{alg.MEIPExperimentAndEstimate} demands us to run a certain system with parameter $\beta_i$ small and wait until convergence. In practice, determining when the system has converged can be done by measuring the output or its derivative. Alternatively, one can use a Lypaunov-based approach \cite{Sharf2019f}. One could also use engineering intuition and simulations to conclude an estimate on the termination time of the experiments. The parameter $\beta_i$ can be chosen in a similar manner.
\end{rem}
\textcolor{black}{
\begin{rem} \label{rem.BetterExperiments}
One can run more experiments to give tighter estimates of $m_i$. Indeed, by construction, take a collection $\{(U_{i,k},Y_{i,k})\}_{k=0}^{r+1}$ of steady-state input-output pairs, and use the monotonicity of the steady-state relation to sort them in such a way that $Y_{i,0} \le \cdots \le Y_{i,r} \le y^\star_i \le Y_{i,r+1}$, $U_{i,0}\le 0 \le U_{i,1} \le \cdots \le U_{i,r+1}$. We would like to use Proposition \ref{prop.estimate_m}, but we do not know the exact values of $y_{i,r},u_{i,0}$. Instead, we again use the monotonicity of the steady-state relation and bound $U_{i,r}\le u_i^\star \le U_{i,r+1}$ and $Y_{i,0}\le y_{i,0} \le Y_{i,1}$. Thus:\small
\begin{align}\label{eq.BetterEstimate}
m_i \le M_i \triangleq \sum_{k=1}^{r} U_{i,k}(Y_{i,k}-Y_{i,k-1}) + U_{i,r+1}(\mathrm y^\star_i - Y_{i,r}).
\end{align}\normalsize
We claim that $M_i$ is a relatively tight estimate of $m_i$. Indeed:
\end{rem}
\begin{prop}
Let $\Delta_k = Y_{i,k}-Y_{i,k-1} \ge 0$, and assume that $k_i^{-1}$ is an $L$-Lipschitz function. Then $|M_i - m_i| \le C\max\{L,1\} \max_{k} \Delta_k$, for a constant $C=C(\mathrm y_{i,0},\mathrm u_i^\star,\mathrm y_i^\star)$.
\end{prop}
\begin{proof}
First, $m_i = K_i^\star(\mathrm y_i^\star) - K_i^\star(\mathrm y_{i,0})=\int_{\mathrm y_{i,0}}^{\mathrm y_i^{\star}} k_i^{-1}(s)ds$. Thus, it is enough to bound each of the following terms:
\begin{enumerate}
\item[i)] $|\int_{Y_{i,k-1}}^{Y_{i,k}} k^{-1}_i(s)ds - U_{i,k}(Y_{i,k}-Y_{i,k-1})|, k=2,\cdots,r$.
\item[ii)] $|\int_{Y_{i,r}}^{\mathrm y^\star_i} k^{-1}_i(s)ds - u_i^\star(\mathrm y^\star_i-Y_{i,r})|$.
\item[iii)] $|\int_{\mathrm y_{i,0}}^{Y_{i,1}} k^{-1}_i(s)ds - U_{i,1}(Y_{i,1}-\mathrm y_{i,0})|$.
\item[iv)] $|(U_{i,r+1}-u_i^\star) (\mathrm y^\star_i-Y_{i,r})|$.
\item[v)] $|U_{i,r+1} (Y_{i,r+1} -\mathrm  y^\star_i)|$
\item[vi)] $|U_{i,1}(\mathrm y_{i,0} - Y_{i,0})|$
\end{enumerate}
The first term can be bounded by:\small
\begin{align*}
\int_{Y_{i,k-1}}^{Y_{i,k}} |k^{-1}(s) - k^{-1}(Y_{i,k})|ds  \le L \int_{Y_{i,k-1}}^{Y_{i,k}}|s-Y_{i,k}|ds = L\Delta_k^2.
\end{align*}\normalsize
Similarly, the second term is bounded by $L(y_i^\star - Y_{i,r})^2$ and the third term is bounded by $L(Y_{i,1}-\mathrm y_{i,0})^2$. The sum of the three terms is thus bounded by:\small
\begin{align*}
&L\left[\sum_{k=2}^r(Y_{i,k}-Y_{i,k-1})^2 + (y_i^\star - Y_{i,r})^2 + (Y_{i,1}-\mathrm y_{i,0})^2\right]\\\le&
L\max(\Delta_k)\left[\sum_{k=2}^r(Y_{i,k}-Y_{i,k-1}) + (y_i^\star - Y_{i,r}) + (Y_{i,1}-\mathrm y_{i,0})\right]\\\le&
L\max(\Delta_k) (\mathrm y_i^\star - \mathrm y_{i,0}),
\end{align*}\normalsize
where we use $y_i^\star - Y_{i,r} \le \Delta_{r+1}$ and $Y_{i,1}-\mathrm y_{i,0} \le \Delta_1$. Similarly, the fourth term is bounded by $L\Delta_{r+1}^2$, the fifth term is bounded by $U_{i,r+1}\Delta_{r+1}$ and the last term is bounded by $U_{i,1}\Delta_{0}$. This completes the proof, as $0\le U_{i,1}\le u_i^\star=k^{-1}_i (\mathrm y_i^\star)$, and $U_{i,r+1} = k^{-1}(Y_{i,r+1}) \le k^{-1}(\mathrm y_i^\star) + L\Delta_{r+1}$.
\end{proof}
\begin{rem}
A natural question that arises is how to conduct experiments assuring that $\max_k \Delta_k$ is small. If we run the system in Fig.~\ref{fig.experiments} with $\beta_i \gg 1$, then the steady-state output would be very close to $y_{\rm ref}$. Thus, if we run $r$ additional experiments with $\beta_i \gg 1$ and references $y_{\rm ref,1}\le y_{\rm ref,2}\le \cdots \le y_{\rm ref,r}$ (i.e., a total of $r+3$ experiments), then $\max \Delta_k = O(\max_k |y_{\rm ref,k+1}-y_{\rm ref,k}|)$. Thus, choosing $y_{\rm ref,k}$ as $r$ equally spaced points in an appropriate interval would give $\max_k \Delta_k = O(1/r)$, and a uniformly random choice gives $\max_k \Delta_k = O(\log r / r)$ with high probability \cite{Holst1980}.
\end{rem}
}
We saw that $m_i$ can be bounded using three experiments for general MEIP agents, \textcolor{black}{and that additional measurements can be used to provide more accurate estimates of $m_i$. Other recent works about data-driven control focus on the case of LTI systems, using them as a base to build toward a solution for nonlinear systems \cite{Recht2019,Coulson2018,DePersis2019}. If we restrict ourselves to this case, we can exactly compute $m_i$ from a single experiment:}
\begin{prop}\label{prop.LinearSampling}
Suppose that the agent $\Sigma_i$ is known to be both MEIP and LTI. Let $(\mathrm{\tilde u,\tilde y})$ be any steady-state input-output pair for which either $\tilde{\mathrm u} \neq 0$ or $\tilde{\mathrm y} = 0.$\footnote{e.g., by running the system in Fig.~\ref{fig.experiments} with some $\beta >0$ and $\mathrm y_{\mathrm{ref}} \ne 0$} Then $m_i = \frac{\mathrm (\mathrm y_i^\star)^2\mathrm{\tilde u}}{2\mathrm{\tilde{y}}}$. Thus $m_i$ can be exactly calculated using a single experiment. 
\end{prop}
\begin{proof}
We know that $k$ is a linear function, and the system state matrix is Hurwitz \cite{Hines2011,Sharf2018a}. Moreover, unless the transfer function of the agent is $0$, $k^{-1}$ is a linear function $k^{-1}(\mathrm y) = s\mathrm y$ for some $s>0$ \cite{Sharf2018b}. Thus $K^\star(\mathrm y) = \frac{s}{2}\mathrm y^2$. 
Now, $k^{-1}(0) = s\cdot 0 = 0$, so $(0,0)$ is a steady-state input-output pair, meaning that $\mathrm y_{i,0} = 0$. Moreover, we know that $\mathrm{\tilde{u}} = s \mathrm{\tilde{y}}$, and not both are zero, so we conclude that $\mathrm{\tilde{y}} \neq 0$, and that $s = \frac{\mathrm{\tilde{u}}}{\mathrm{\tilde{y}}}$.
Thus, $K_i^\star(\mathrm y_{i,0}) = K_i^\star(0) = 0$ and $K_i^\star(\mathrm y_i^\star) = \frac{s}{2}(\mathrm y_i^\star)^2$. This completes the proof, as $m_i = K_i^\star(\mathrm y_i^\star) - K_i^\star(\mathrm y_{i,0})$.
\end{proof}

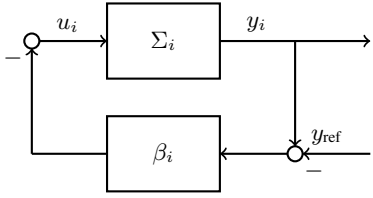
\begin{figure}
\begin{center}
\begin{tikzpicture}
\draw[thick, fill=white!10] (0,0) rectangle (1.5,1);
\draw[thick] (0,-1.5) rectangle (1.5,-.5);
\draw[thick,->] (1.5,0.5) -- (3.5,0.5);
\node[anchor=south east] at (2.25,0.5) {\small $y_{i}$};
\node[anchor=south east] at (-.25,0.5) {\small $u_{i}$};
\node[anchor=south east] at (3.25,-1) {\small $y_{\text{ref}}$};

\draw[thick,-] (0,-1) -- (-1,-1);

\draw[thick] (2.5,-1) circle(0.1);
\draw[thick] (-1,0.5) circle(0.1);
\node[anchor=north east] at (-1,0.5) {\small $-$};
\node[anchor=north west] at (2.5,-1) {\small $-$};

\draw[thick,->] (3.5,-1) -- (2.6,-1);
\draw[thick,->] (2.5,0.5)--(2.5,-0.9);
\draw[thick,->] (-1,-1) -- (-1,0.4);
\draw[thick,->] (-0.9,0.5) -- (0,0.5);
\draw[thick,->] (2.4,-1) -- (1.5,-1);

\node at (0.75,0.5) {\small $\Sigma_{i}$};
\node at (0.75,-1) {\small $\beta_i$};

\end{tikzpicture}
\end{center}
\caption{Experimental setup of the closed-loop experiment for estimating $m_i$ as used in Algorithm~\ref{alg.MEIPExperimentAndEstimate}.}
\label{fig.experiments}
\end{figure}

The chapter concludes with Algorithm \ref{alg.SingleGain} for solving the practical formation control problem using the single-gain amplification scheme, which is applied in Section \ref{sec.Simulations}.

\begin{algorithm} 
\caption{Synthesis Procedure for Practical Formation Control}
\label{alg.SingleGain}
\SetAlgoLined  
\small Choose some output-strictly MEIP controllers $\Pi_e$ such that the integral function $\Gamma$ has a single minimizer $\upzeta^\star$ when restricted to the set $\IM(\E^T)$\; 
Choose some $\mathrm y^\star \in \mathbb{R}^n$ such that $\E^T\mathrm y^\star=\upzeta^\star$\;
\For{$i = 1,...,n$}
{
Run Algorithm \ref{alg.MEIPExperimentAndEstimate}. Let $\mathbb{m}_i$ be the output\;
}
Let $\mathbb{m} = \sum_{i=1}^{n} \mathbb{m}_i$\;
Compute $M = \min\{\zeta\in\IM(\E^T):\ \|\zeta-\upzeta^\star\| = \varepsilon\}$\;
Compute $\alpha = \mathbb{m}/M$\;
\textbf{return} the controllers $\{\alpha\Pi_e\}_{e\in\EE}$;
\end{algorithm}

\begin{rem}
Step 1 of the algorithm allows almost complete freedom of choice for the controllers. One possible choice are the static controllers $\mu_e = \zeta_e - \upzeta_e^\star$. Moreover, if $\Pi_e$ is any MEIP controller for each $e\in\EE$, and $\gamma_e(\zeta_e) = 0$ has a unique solution for each $e\in\EE$, then the ``formation reconfiguration" scheme from \cite{Sharf2017} suggests a way to find the required controllers using mild augmentation.
\end{rem}
\begin{rem}
The algorithm allows one to choose any vector $\mathrm y^\star$ such that $\E^T \mathrm y^\star = \upzeta^\star$. All possible choices lead to some gain $\alpha$ which assures a solution of the practical formation control problem, but some choices yield better results (i.e., smaller gains) than others. The optimal $\mathrm y^\star$, minimizing the estimate $\alpha$, can be found as the minimizer of the problem $\min \{K^\star(\mathrm y):\ \E^T\mathrm y = \upzeta^\star\}$, which we cannot compute using data alone. One can use physical intuition to choose a vector $\mathrm y^\star$ which is relatively close to the actual minimizer, but the algorithm is still valid no matter which $\mathrm y^\star$ is chosen.
\end{rem}

\vspace{-10pt}
\section{Iterative Practical Formation Control: Applying Different Gains on Different Edges} \label{sec.multigain}
Let us revisit Fig.~\ref{fig.ClosedLoopGain} and let $A = \diag(\{a_e\}_{e\in\EE}) $ with positive, but distinct entries $a_e$. These additional degrees of freedom can be used, for example, to reduce the conservatism and retrieve a smaller norm of the adjustable gain vector $a$ while still solving the practical formation control problem. It follows directly from Theorem~\ref{thm.practical_formation_control_coupling} that there always exists a bounded vector $a$ solving the practical formation control problem. However, the question remains how $a$ can be chosen based on sampled input-output data and passivity properties.

Our idea here is to probe our diffusively coupled system for given gains $a_e$ and adjust the gains according to the resulting steady-state output. By iteratively performing experiments in this way, we strive to find controller gains that solve the practical formation control problem. This approach is tightly connected to iterative learning control, where one iteratively applies and adjusts a controller to improve the performance of the closed-loop for a repetitive task \cite{Bristow2006}. Our approach here is based on passivity and network optimization with only requiring the possibility to perform iterative experiments.

One natural idea in this direction is to define a cost function that penalizes the distance of the resulting steady-state to the desired formation control goal and then apply a gradient descent approach, adjusting the gain $a$ for each experiment. 
However, to obtain the gradient of $\| \E^T \mathrm y(a) - \upzeta^\star \|^2$  with respect to the vector $a$, where $ \mathrm y (a)$ is the steady-state output of $(\G,\Sigma,\Pi,a)$, one requires knowledge of the inverse relations $k_i^{-1}$ for all $i=1, \dots, n$. With no model of the agents available, a direct gradient descent approach is hence infeasible. 
We thus look for a simple iterative multi-gain control scheme without knowledge on the exact steepest descent direction. 

We start off with an arbitrarily chosen gain vector $a_0$ with only positive entries. Due to Assumption \ref{Assump}, the closed-loop converges to a steady state. According to the measured state, the idea is then to iteratively perform experiments and update the gain vector until we reach our control goal, i.e., practical formation control. The update formula can be summarized by
\begin{align}
a_e^{(j+1)} = a_e^{(j)} + h v_e , \quad e\in \EE,
\label{eq:multigain1}
\end{align}
where $h>0$ is the step size and $v$, with entries $v_e$, $e=1, \dots, |\EE|$, is the update direction. \textcolor{black}{In practice, the update can either be instantaneous or gradual, e.g. using linear interpolation or higher-order splines.}
We denote the $e$-th entry of $\E^T \mathrm y$ as $\mathrm f_e$ and choose $v$ in each iteration such that 
\begin{align}
v_e = \begin{cases}\frac{\mathrm f_e -\upzeta^\star_e}{\gamma_e(\mathrm f_e)} & \quad \gamma_e(\mathrm f_e) \neq 0 \\ 0 & \quad \text{otherwise} \end{cases}, \quad e\in\mathbb{E}.
\label{eq:multigain2}
\end{align}
If $k^{-1}$ and $\gamma$ are differentiable functions, then we claim that $F(a) = ||\E^T \mathrm y (a) - \upzeta^\star||^2$ decreases in the direction of $v$, i.e., $v^T \nabla F(a)< 0$.
This leads to a multi-gain distributed control scheme, using \eqref{eq:multigain1} with \eqref{eq:multigain2}, summarized in Algorithm~\ref{alg.MultiGain2}. This multi-gain distributed control scheme is guaranteed to solve the practical formation problem after a finite number of iterations, and is summarized in the following theorem.

\begin{algorithm}[!h]
 \caption{Practical Formation control with derivative-free optimization} 
\label{alg.MultiGain2}
\SetAlgoLined
\small Initialize $a^{(0)}$, e.g. with $\mathbb{1}_{| \EE |}$ \;
Choose step size $h$ and set $j=0$\;
\While{$ F(a) = \|\E^T \mathrm y (a) - \upzeta^\star\|^2>\varepsilon$}{
Apply $a^{(j)}$ to the closed loop \;
Compute $v_e = \begin{cases}\frac{\mathrm f_e -\upzeta^\star_e}{\gamma_e(\mathrm f_e)} &\gamma(\mathrm f_e) \neq 0 \\ 0 & \gamma(\mathrm f_e) = 0\end{cases} \; \forall e$ \; 
Update $a_{e}^{(j+1)} = a_e^{(j)} + h v_e$, $j = j+1$ \;
}
\textbf{return} $a$ 
\end{algorithm}

\begin{thm} \label{thm.DerivativeFreeConvergence}
Suppose that the functions $k^{-1}$, $\gamma$ are differentiable, and that there exists an agent $i_0\in\V$ such that $\frac{dk^{-1}_{i_0}}{d\mathrm y_{i_0}} > 0$ for any point $\mathrm y_{i_0}\in \mathbb{R}$. Moreover, assume that $\frac{d\gamma_e}{d\mathrm \zeta_e} > 0$ for any $e\in\EE, \zeta_e\in \mathbb{R}$. 
Then $v^T\nabla F(a) \le 0,$ with $v, F$ as defined in Algorithm \ref{alg.MultiGain2} (and equality if and only if $\E^T\mathrm y (a) = \upzeta^\star$).
Furthermore, if the step size $h>0$ is small enough, then the Algorithm \ref{alg.MultiGain2} halts after finite time, providing a gain vector that solves the practical formation control problem.
\end{thm}

\begin{proof}[Sketch of proof]
The proof is based on showing that $\nabla F$ can be written as $-\mathrm{diag}(\gamma(\mathrm f)) X(\mathrm y(a)) (\mathrm f - \upzeta^\star)$, where $X(\mathrm y(a))$ is a positive-definite matrix depending on $\mathrm y(a)$. We can now show that $v^T \nabla F = - (\mathrm f - \upzeta^\star)^TX(\mathrm y(a)) (\mathrm f - \upzeta^\star) \le 0$. The full proof of Theorem \ref{thm.DerivativeFreeConvergence} is available in the appendix.
\end{proof}

Algorithm~\ref{alg.MultiGain2} together with the theoretical results from Theorem \ref{thm.DerivativeFreeConvergence} provide us with a very simple and distributed, iterative control scheme with theoretical guarantees. Note also, that the steady-states of the agents are independent of their initial condition. For each iteration, the agents can hence also start from the position they converged to at the last iteration. This can be interpreted similarly to Remark \ref{rem.data_free}, where gains are updated on a slower time scale than convergence of the agents. However, instead of only the information whether practical formation control is achieved, we generally need the actual difference $\E^T \mathrm y - \upzeta^\star$ that is achieved with the current controller in each iteration. In the special case of proportional controllers $\mu_e = \zeta_e - (\upzeta^\star)_e$, yielding $v_e = 1$, we retrieve the exact controller scheme proposed in Remark \ref{rem.data_free}.

An alternative gradient-free control scheme is the extremum seeking framework presented in \cite{Feiling2018}. Assuming that $k^{-1}$ and $\gamma$ are twice continuously differentiable, a step in the direction of steepest descent is approximated every $4 |\EE|$ steps (cf.~\cite[Theorem 1]{Feiling2018}). While the extremum seeking framework approximates the steepest descent (and the simple multi-gain approach only guarantees a descending direction), it also requires large amounts of experiments per approximated gradient step. Furthermore, the algorithm as presented in \cite{Feiling2018} cannot be computed in a purely distributed fashion. Therefore, the simple distributed control scheme in Algorithm~\ref{alg.MultiGain2} displays significant advantages in the present problem setup.
\vspace{-7pt}
\section{Case Study: Velocity Coordination in Vehicles with Drag and Exogenous Forces} \label{sec.Simulations}
Consider a collection of {$30$} one-dimensional robot vehicles, each modeled by a double integrator $G(s) = \frac{1}{s^2}$. The robots try to coordinate their velocity. Each of them has its own drag profile $f(\dot{p})$, which is unknown to the algorithm, but it is known that $f$ is increasing and $f(0)=0$. Moreover, each vehicle experiences external forces (e.g., wind, and being placed on a slope). The velocity of the vehicles is governed by the equation
\begin{align}
\Sigma_i:\ \begin{cases}\dot{x_i} &= -f_i(x_i) + u_i + w _i \\ y_i &= x_i, \end{cases}
\end{align}
where $x_i$ is the velocity of the $i$-th vehicle, $f_i$ is its drag model, $w_i$ are the exogenous forces acting on it, $u_i$ is the control input, and $y_i$ is the measurement. In the simulation, the drag models $f_i$ are given by $c_d |x|x$, where the drag coefficient $c_d$ is chosen as a log-uniformly distributed random variable. We assume that the vehicles are light, so the wind accelerates the vehicles by a non-negligible amount. Thus, $w_i$ is randomly chosen between $-2$ and $2$. 
We wish to achieve velocity consensus, with error no greater than $\epsilon = 0.2$. We consider a diffusive coupling of the agents with the cycle graph $\G=\mathcal{C}_{30}$, and take proportional controllers $\zeta_e = \mu_e$.

We apply the amplification scheme presented in Algorithm \ref{alg.SingleGain} and choose the consensus value $y^\star_i = 1.5_{m/sec}$ to use in the estimation algorithm. Note that the plants are MEIP, but not output-strictly MEIP, and use Algorithm \ref{alg.MEIPExperimentAndEstimate} to estimate the required uniform gain $\alpha$. The first two experiments are conducted with $\beta_i = 0.01$, and $y_{\Refc} = \pm 100$. Based on their results, we run a third experiment on each of the agents for which this is required, this time with $\beta_i = 1$ and $y_{\Refc} = \pm 10$, where the sign is chosen according to Algorithm \ref{alg.MEIPExperimentAndEstimate}. The experimental results are available in Figure \ref{fig.CaseStudy1Experiments}. 

We estimate each $m_i$ using Remark \ref{rem.estimate_m_2pt}. For example, for agent 1 we get the three steady-state input-output pairs $(0.9947,0.5203)$, $(-0.9687,-3.1294)$, and $(3.4268,3.5732)$. Monotonicity implies that it has steady-states $(u^\star_1,y^\star_1)=(u^\star_1,1.5)$ and $(0,y_{1,0})$ with $0.9947 \le u^\star_1 \le 3.4268$ and $-3.1294 \le y_{1,0} \le 0.5203$. We can thus estimate $m_2 \le 3.4268 \cdot (1.5 - (-3.1294)) = 15.864$. Repeating this calculation for each of the agents and summing gives $m = 256.3658$.

\begin{figure*}[t!]
	\centering
	\subfigure[Results of First Set of Experiments for the Vehicles. Each plot corresponds to a different agent.]{\includegraphics[scale=0.65]{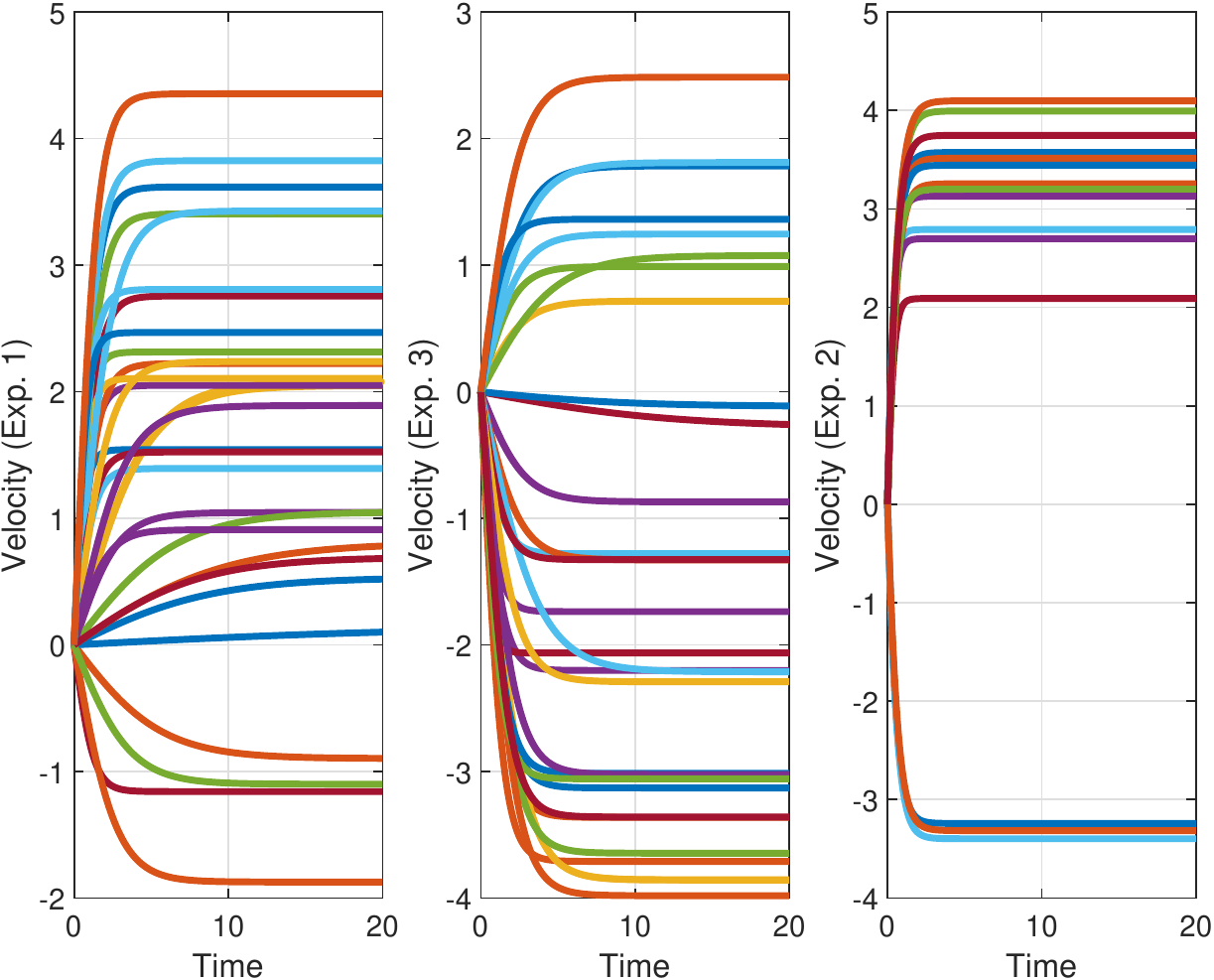}\label{fig.CaseStudy1Experiments}
}\hspace{0.1cm}
	\subfigure[The Closed-Loop with Uniform Gain $\alpha = 213.638$. The two leftmost graphs plot the agents' trajectories over 0.3 seconds and over 10 seconds. The rightmost graph plots the relative outputs.]{\includegraphics[scale=0.65]{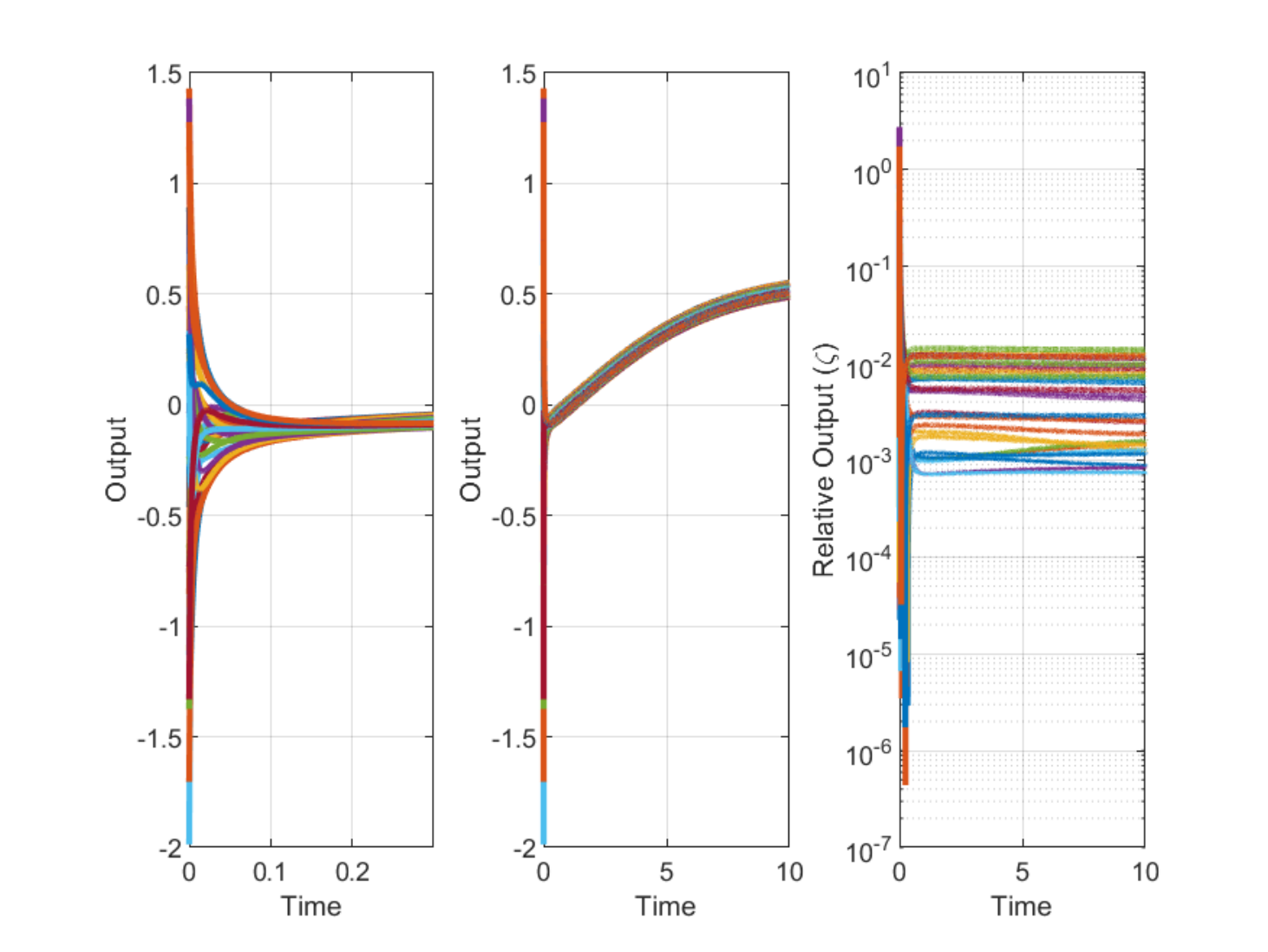}\label{fig.CaseStudy1EndRun}}
	\vspace{-8pt}
	\caption{Experiment results for vehicle case study.}
	\vspace{-5pt}
\end{figure*}

As for estimating $M$, we have $\Gamma_e(\zeta_e) = \zeta^2_e$, so $\Gamma(\zeta) = \|\zeta\|^2$. The minimum is at $\zeta^\star = 0$, and by definition we have $M = \min_{\zeta\in\IM(\E^T):\ \|\zeta-\upzeta^\star\| =  \varepsilon} \Gamma(\zeta) = |\EE|\varepsilon^2$. Thus, we get $\alpha = \frac{m}{M} = \frac{m}{30\varepsilon^2} = 213.64$. To verify the algorithm, we run the closed-loop system $(\G,\Sigma,\Pi,\alpha\mathbb{1})$ with the gain $\alpha$ we found. The results are available in Figure \ref{fig.CaseStudy1EndRun}. One can see that the overall control goal is achieved - the agents converge to a steady-state which is $\varepsilon$-close to consensus.  However, the agents actually converge to a much closer steady-state than predicted by the algorithm. Namely, the distance of the steady-state output from consensus is roughly $0.04$, much smaller than $0.2$. \textcolor{black}{Uncoincidentally, the true value of $m=K(0)+K^\star(y^\star)$ is $50.15$, meaning we overestimate it (and hence $m/M$) by about $411\%$. One can mitigate this by using more experiments to improve the estimate $m_i$, as in Proposition \ref{prop.estimate_m} or in Remark \ref{rem.BetterExperiments}. We follow this approach and conduct further experiments on each of the agents using $\beta_i = 10$ and choosing $y_{\Refc}$ randomly. The resulting values of $\alpha$, as well as the error from the true value of $m/M$, can be seen in the table below. It can be seen that even a single additional measurement per agent can significantly reduce the estimation error of $m/M$. 
\begin{center}
{\scriptsize
    \begin{tabular}{ | p{80pt} || p{40pt} | p{75pt} |}
    \hline
    Measurements Per Agent& Value of $\alpha$ & Overestimation of $m/M$ \\ \hline
    	3 & $213.638$ & $411\%$ \\
    	4 & $104.308$ & $150\%$ \\
    	10 & $67.316$ & $61\%$ \\
   	20 & $55.161$ & $32\%$ \\  \hline
    \end{tabular}
    }
\end{center}}

Altogether we showed that Algorithm \ref{alg.SingleGain} manages to solve the practical consensus problem for vehicles, affected by drag and exogenous inputs, without using any model for the agents, while conducting very few experiments for each agent. However, it overestimates the required coupling, and thus has unnecessarily large energy consumption.

Let us now apply the iterative multi-gain control strategy. We start with $a^{(0)} = 0.1 \otimes 1_{|\EE|}$, we choose the step size $h=2$ and apply Algorithm~\ref{alg.MultiGain2}. In fact, since $\upzeta^\star = 0$ and $\zeta_e = \mu_e$, we receive $v=1_{|\EE|}$, which constitutes the special case of Remark \ref{rem.data_free}. The corresponding norm of the gain vector and the resulting $\varepsilon$ in each iteration is illustrated in Fig.~\ref{fig.CaseStudy1_MG2_a}. After $20$ iterations, we already arrive at a vector, which solves the practical formation problem with $\|a^{(19)}\| = 208.68$, while $\varepsilon = 0.195 < 0.2$.  Note that the controller with the uniform gain had $\|a\| = \sqrt{|\EE|}\cdot 213.638 = 1170.1$, so the iterative scheme beats it by a factor of $5$ in terms of energy.

\begin{figure} 
    \centering
%
%
%
\begin{tikzpicture}
 \tikzstyle{every node}=[font=\scriptsize]
\begin{semilogyaxis}[%
width=2.3cm,
height=2.3cm,
at={(0cm,0cm)},
scale only axis,
xmin=0,
xmax=20,
xlabel={Iterations $j$},
xmajorgrids,
ymin=0,
ymax=14,
ylabel={$\varepsilon (a^{(j)})$},%
ylabel near ticks,
ymajorgrids,
axis background/.style={fill=white},
legend style={legend cell align=left,align=left,draw=white!15!black}
]
\addplot [color=red,mark size=1pt,only marks,mark=*,mark options={solid},forget plot]
  table[row sep=crcr]{%
0   1.221982415716640e+01\\
1   2.337842261964685e+00\\
2   1.409799463333590e+00\\
3   1.014581918017397e+00\\
4   7.872803801071873e-01\\
5   6.424998997169324e-01\\
6   5.464120195751423e-01\\
7   4.772780788506241e-01\\
8   4.245738984408815e-01\\
9   3.828348153326068e-01\\
10   3.488551434559351e-01\\
11   3.205984707163623e-01\\
12   2.966976435752313e-01\\
13   2.761961160724759e-01\\
14   2.584032384737470e-01\\
15   2.428062742733395e-01\\
16   2.290158205882989e-01\\
17   2.167307692963437e-01\\
18   2.057142229402757e-01\\
19   1.957769519701794e-01\\
20   1.867662743750864e-01\\
};
\end{semilogyaxis}

\begin{axis}[%
width=2.3cm,
height=2.3cm,
at={(4cm,0cm)},
scale only axis,
xmin=0,
xmax=20,
xlabel={Iterations $j$},
xmajorgrids,
ymin=0,
ymax=230, 
ylabel={$\|a_j\|$},%
ylabel near ticks,
ymajorgrids,
axis background/.style={fill=white},
legend style={legend cell align=left,align=left,draw=white!15!black}
]
\addplot [color=blue,mark size=1pt,only marks,mark=*,mark options={solid},forget plot]
  table[row sep=crcr]{%
0   5.477225575051662e-01\\
1   1.150217370760849e+01\\
2   2.245662485771181e+01\\
3   3.341107600781513e+01\\
4   4.436552715791845e+01\\
5   5.531997830802177e+01\\
6   6.627442945812510e+01\\
7   7.722888060822842e+01\\
8   8.818333175833175e+01\\
9   9.913778290843507e+01\\
10   1.100922340585384e+02\\
11   1.210466852086417e+02\\
12   1.320011363587450e+02\\
13   1.429555875088484e+02\\
14   1.539100386589517e+02\\
15   1.648644898090550e+02\\
16   1.758189409591583e+02\\
17   1.867733921092617e+02\\
18   1.977278432593650e+02\\
19   2.086822944094683e+02\\
20   2.196367455595716e+02\\
};
\end{axis}
\end{tikzpicture}%
    \caption{The resulting $\varepsilon$ and the norm of the gain vector $\|a^{(j)}\|$ over iterations $j$ when applying the iterative multi-gain control strategy to the case study of velocity coordination in vehicles.}
		\label{fig.CaseStudy1_MG2_a}
		\vspace{-15pt}
\end{figure}
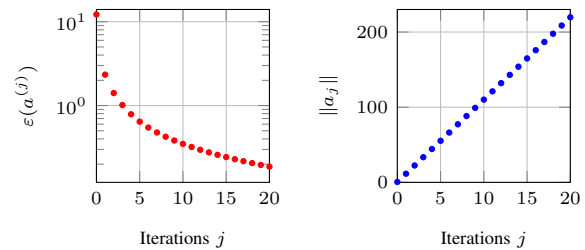
\vspace{-10pt}
\section{Conclusions and Future Work} \label{sec.Conclusions}
We presented an approach for model-free practical cooperative control for diffusively coupled systems only on the premise of passivity of the agents. The presented approach led to two control schemes: with additional two or three experiments on the agents, we can upper bound the controller gain which solves the practical formation problem, or we can iteratively adapt the adjustable gain vector until practical formation is reached. Both approaches are especially simple in their application, while still being scalable and providing theoretical guarantees.
Future research might try and improve the presented methods, either by reducing the number of experiments needed on each agent, or by achieving faster practical convergence using iterations. One might also try to use very limited knowledge on the agents to achieve the said improvement. Other possible directions include data-driven solutions to more intricate problems using the network optimization framework, e.g. fault detection and isolation.
\vspace{-10pt}


\bibliographystyle{ieeetr}
\bibliography{main}

\begin{thebibliography}{10}

\bibitem{Franchi2011}
A.~Franchi, P.~R. Giordano, C.~Secchi, H.~I. Son, and H.~H. Bulthoff,
  ``Passivity-based decentralized approach for the bilateral teleoperation of a
  group of uavs with switching topology,'' in {\em IEEE International
  Conference on Robotics and Automation}, pp.~898--905, 2011.

\bibitem{Bando1995}
M.~Bando, K.~Hasebe, A.~Nakayama, A.~Shibata, and Y.~Sugiyama, ``Dynamical
  model of traffic congestion and numerical simulation,'' {\em Phys. Rev. E},
  vol.~51, pp.~1035--1042, Feb 1995.

\bibitem{Urban2001}
D.~Urban and T.~Keitt, ``Landscape connectivity: A graph-theoretic
  perspective,'' {\em Ecology}, vol.~82, no.~5, pp.~1205--1218, 2001.

\bibitem{Recht2019}
B.~Recht, ``A tour of reinforcement learning: The view from continuous
  control,'' {\em Annual Review of Control, Robotics, and Autonomous Systems},
  vol.~2, no.~1, pp.~253--279, 2019.

\bibitem{Dean2017}
S.~Dean, H.~Mania, N.~Matni, B.~Recht, and S.~Tu, ``On the sample complexity of
  the linear quadratic regulator,'' {\em Foundations of Computational
  Mathematics}, 10 2017.

\bibitem{Gorges2019}
D.~G\"orges, ``Distributed adaptive linear quadratic control using distributed
  reinforcement learning,'' {\em IFAC-PapersOnLine}, vol.~52, no.~11, pp.~218
  -- 223, 2019.
\newblock 5th IFAC Conference on Intelligent Control and Automation Sciences
  ICONS 2019.

\bibitem{Coulson2018}
J.~Coulson, J.~Lygeros, and F.~D{\"o}rfler, ``Data-enabled predictive control:
  In the shallows of the {DeePC},'' in {\em 18th European Control Conf.},
  pp.~307--312, 2019.

\bibitem{Fattahi2019}
S.~Fattahi, N.~Matni, and S.~Sojoudi, ``Efficient learning of distributed
  linear-quadratic controllers,'' {\em arXiv preprint arXiv:1909.09895}, 2019.

\bibitem{Jiang2018}
H.~Jiang and H.~He, ``Data-driven distributed output consensus control for
  partially observable multiagent systems,'' {\em IEEE Trans.\ Cybernetics},
  pp.~1--11, 2018.

\bibitem{Bianchini2017}
F.~Bianchini, G.~Fenu, G.~Giordano, and F.~A. Pellegrino, ``Model-free tuning
  of plants with parasitic dynamics,'' in {\em Proc.\ 56th IEEE Conf.\ on
  Decision and Control}, pp.~499--504, Dec 2017.

\bibitem{Siljak1996}
D.~D. \v{S}iljak, ``Decentralized control and computations: Status and
  prospects,'' {\em A.\ Rev.\ Control}, vol.~20, 1996.

\bibitem{Zheng2018}
Y.~Zheng, S.~E. Li, K.~Li, and W.~Ren, ``Platooning of connected vehicles with
  undirected topologies: Robustness analysis and distributed {H}-infinity
  controller synthesis,'' {\em IEEE Trans.\ Intelligent Transportation
  Systems}, vol.~19, pp.~1353--1364, 2018.

\bibitem{Yu2012}
W.~Yu, P.~DeLellis, G.~Chen, M.~di~Bernardo, and J.~Kurths, ``Distributed
  adaptive control of synchronization in complex networks,'' {\em IEEE Trans.\
  Automat.\ Control}, vol.~57, pp.~2153--2158, 2012.

\bibitem{Montenbruck2016}
J.~M. Montenbruck and F.~Allg{\"o}wer, ``Some problems arising in controller
  design from big data via input-output methods,'' in {\em Proc.\ 55th IEEE
  Conf.\ on Decision and Control}, pp.~6525--6530, 2016.

\bibitem{Romer2017b}
A.~Koch, J.~M. Montenbruck, and F.~Allgower, ``Sampling strategies for
  data-driven inference of input-output system properties,'' {\em IEEE Trans.\
  Automat.\ Control}, 2021.

\bibitem{Romer2019}
A.~Romer, J.~Berberich, J.~K{\"o}hler, and F.~Allg{\"o}wer, ``One-shot
  verification of dissipativity properties from input-output data,'' {\em IEEE
  Control Systems Letters}, vol.~3, pp.~709--714, 2019.

\bibitem{Khalil2001}
H.~K. Khalil, {\em Nonlinear Systems}.
\newblock Pearson, 3rd~ed., 2001.

\bibitem{Arcak2007}
M.~Arcak, ``Passivity as a design tool for group coordination,'' {\em IEEE
  Trans.\ Automat.\ Control}, vol.~52, pp.~1380--1390, Aug. 2007.

\bibitem{Bai2011}
H.~Bai, M.~Arcak, and J.~Wen, {\em Cooperative Control Design: A Systematic,
  Passivity-Based Approach}.
\newblock Communications and Control Engineering, Springer, 2011.

\bibitem{Pavlov2008}
A.~Pavlov and L.~Marconi, ``Incremental passivity and output regulation,'' {\em
  Systems \& Control Letters}, vol.~57, no.~5, pp.~400 -- 409, 2008.

\bibitem{Hines2011}
G.~H. Hines, M.~Arcak, and A.~K. Packarda, ``Equilibrium-independent passivity:
  A new definition and numerical certification,'' {\em Automatica}, vol.~47,
  no.~9, pp.~1949--1956, 2011.

\bibitem{Burger2014}
M.~B{\"u}rger, D.~Zelazo, and F.~Allg{\"o}wer, ``Duality and network theory in
  passivity-based cooperative control,'' {\em Automatica}, vol.~50, no.~8,
  pp.~2051--2061, 2014.

\bibitem{Sharf2017}
M.~Sharf and D.~Zelazo, ``A network optimization approach to cooperative
  control synthesis,'' {\em IEEE Control Systems Letters}, vol.~1, pp.~86--91,
  July 2017.

\bibitem{Sharf2018a}
M.~{Sharf} and D.~{Zelazo}, ``Analysis and synthesis of mimo multi-agent
  systems using network optimization,'' {\em IEEE Transactions on Automatic
  Control}, vol.~64, no.~11, pp.~4512--4524, 2019.

\bibitem{Jain2018}
A.~Jain, M.~Sharf, and D.~Zelazo, ``Regulatization and feedback passivation in
  cooperative control of passivity-short systems: A network optimization
  perspective,'' {\em IEEE Control Systems Letters}, vol.~2, pp.~731--736,
  2018.

\bibitem{Sharf2019a}
M.~Sharf, A.~Jain, and D.~Zelazo, ``A geometric method for passivation and
  cooperative control of equilibrium-independent passivity-short systems,''
  {\em arXiv preprint arXiv:1901.06512}, 2019.

\bibitem{Sharf2019f}
M.~Sharf and D.~Zelazo, ``A data-driven and model-based approach to fault
  detection and isolation in networked systems,'' {\em arXiv preprint
  arXiv:1908.03588}, 2019.

\bibitem{Montenbruck2015}
J.~M. Montenbruck, M.~B{\"u}rger, and F.~Allg{\"o}wer, ``Practical
  synchronization with diffusive couplings,'' {\em Automatica}, vol.~53,
  pp.~235 -- 243, 2015.

\bibitem{Kim2016}
J.~{Kim}, J.~{Yang}, H.~{Shim}, J.~{Kim}, and J.~H. {Seo}, ``Robustness of
  synchronization of heterogeneous agents by strong coupling and a large number
  of agents,'' {\em IEEE Trans.\ Automat.\ Control}, vol.~61, no.~10,
  pp.~3096--3102, 2016.

\bibitem{Godsil2001}
C.~Godsil and G.~Royle, {\em Algebraic Graph Theory}.
\newblock Graduate Texts in Mathematics, Springer New York, 2001.

\bibitem{Rockafeller1997}
R.~T. Rockafellar, {\em Convex Analysis}.
\newblock Princeton Landmarks in Mathematics and Physics, Princeton University
  Press, 1997.

\bibitem{Oh2015}
K.-K. Oh, M.-C. Park, and H.-S. Ahn, ``A survey of multi-agent formation
  control,'' {\em Automatica}, vol.~53, pp.~424 -- 440, 2015.

\bibitem{Romer2017a}
A.~Romer, J.~M. Montenbruck, and F.~Allg{\"o}wer, ``Determining dissipation
  inequalities from input-output samples,'' in {\em Proc.\ 20th IFAC World
  Congress}, pp.~7789--7794, 2017.

\bibitem{Holst1980}
L.~Holst, ``On the lengths of the pieces of a stick broken at random,'' {\em
  Journal of Applied Probability}, vol.~17, no.~3, pp.~623--634, 1980.

\bibitem{DePersis2019}
C.~{De Persis} and P.~{Tesi}, ``Formulas for data-driven control:
  Stabilization, optimality and robustness,'' {\em IEEE Trans.\ Automat.\
  Control}, 2019.

\bibitem{Sharf2018b}
M.~Sharf and D.~Zelazo, ``Network identification: A passivity and network
  optimization approach,'' in {\em Proc.\ 57th IEEE Conf.\ on Decision and
  Control}, 2018.

\bibitem{Bristow2006}
D.~A. {Bristow}, M.~{Tharayil}, and A.~G. {Alleyne}, ``A survey of iterative
  learning control,'' {\em IEEE Control Systems Magazine}, vol.~26, no.~3,
  pp.~96--114, 2006.

\bibitem{Feiling2018}
J.~Feiling, A.~Zeller, and C.~Ebenbauer, ``Derivative-free optimization
  algorithms basen on non-commutative maps,'' {\em IEEE Control Systems
  Letters}, vol.~2, pp.~743--748, 2018.

\end{thebibliography}
\vspace{-5pt}
\appendix
We now give full proofs of Proposition \ref{prop.MEIPFromObscureModel} and Theorem \ref{thm.DerivativeFreeConvergence}.
\vspace{-10pt}
\subsection{Proving Proposition \ref{prop.MEIPFromObscureModel}}
In order to prove the proposition, we use the notion of cursive relations established in \cite{Sharf2019a}:
\begin{defn}[Cursive Relations, \cite{Sharf2019a}] \label{defn_cursive_relation}
A set $A \subset \mathbb{R}^2$ is called \emph{cursive} if there exists a curve $\alpha:\mathbb{R}\to\mathbb{R}^2$ such that the following conditions hold:
\begin{itemize}
\item[i)] The set $A$ is the image of $\alpha$.
\item[ii)] The map $\alpha$ is continuous.
\item[iii)] The map $\alpha$ satisfies $\lim\limits_{|t|\to\infty} \|\alpha(t)\| = \infty$.
\item[iv)] $\{t\in\mathbb{R}:\ \exists s\neq t,\ \alpha(s)=\alpha(t)\}$ has measure zero.
\end{itemize}
A relation $\varUpsilon$ is called cursive if the set $\{(p, q):\ q \in \varUpsilon(p)\}$ is cursive.
\end{defn}
The notion of cursive relations is useful as it can help prove that systems are MEIP. Specifically,
\begin{thm}[\hspace{-2pt} \cite{Sharf2019a}]\label{thm.Sharf2019a}
A monotone cursive relation is maximally monotone.
\end{thm}
We can now prove Proposition \ref{prop.MEIPFromObscureModel}:
\begin{proof}
Consider an arbitrary steady-state of the system. As $h$ is continuous and strictly monotone ascending, hence invertible, we must have $\dot{x}=0$ for any steady-state input-output pair. Thus, we conclude that any steady-state input-output pair can be written as $(f(\sigma)/g(\sigma),h(\sigma))$ for some $\sigma \in \mathbb{R}$. We first show passivity with respect to every steady-state, and then show that the steady-state input-output relation is maximally monotone.
Take a steady-state $(f(x_0)/g(x_0),h(x_0))$ of the system, and define $S(x)=\int_{x_0}^x \frac{h(\sigma)-h(x_0)}{g(\sigma)}d\sigma$. We claim that $S$ is a storage function for the steady-state input-output pair $(f(x_0)/g(x_0),h(x_0))$. Indeed, $S(x) \ge 0$, with equality only at $x_0$, immediately follows from strict monotonicity of $h$ and $g > 0$. As for the inequality defining passivity, we have:%

\small
\begin{align*}
&\frac{d}{dt} S(x) =\frac{h(x)-h(x_0)}{g(x)}(-f(x)+g(x)u) =\\
&(h(x)-h(x_0))u - \frac{f(x)}{g(x)}(h(x)-h(x_0)) =\\ 
&(h(x)-h(x_0))\bigg(u-\frac{f(x_0)}{g(x_0)}\bigg) - \bigg(\frac{f(x)}{g(x)}-\frac{f(x_0)}{g(x_0)}\bigg)(h(x)-h(x_0)), 
\end{align*}
\normalsize
where the second term is negative as $\frac{f}{g},h$ are monotone ascending, and the first term is $(y-h(x_0))(u-\frac{f(x_0)}{g(x_0)})$. Hence, the system is indeed passive with respect to any steady-state input-output pair.
As for maximal monotonicity of the steady-state relation, we recall that it can be parameterized as $(f(\sigma)/g(\sigma),h(\sigma))$ for $\sigma \in \mathbb{R}$. We claim that this relation is both monotone and cursive, which will prove maximal monotonicity. Monotonicity holds as for any $x_0,x_1$,\small
\begin{align} \label{eq.MonotonicityExample}
\frac{f(x_0)}{g(x_0)} > \frac{f(x_1)}{g(x_1)} \iff x_0 > x_1 \iff h(x_0) > h(x_1)
\end{align} \normalsize
due to strict monotonicity. As for cursiveness, the map $\sigma \mapsto (f(\sigma)/g(\sigma),h(\sigma))$ is a curve whose image is the relation. Moreover, it is clear that the map is continuous, and also injective due to \eqref{eq.MonotonicityExample}. Lastly, we have
\begin{align}
\lim_{|t|\to\infty} \bigg\|\bigg(\frac{f(t)}{g(t)},h(t)\bigg)\bigg\| \ge \lim_{|t|\to\infty} \max\bigg\{\bigg|\frac{f(t)}{g(t)}\bigg|,|h(t)|\bigg\} = \infty
\end{align}
so the proof is complete by Theorem \ref{thm.Sharf2019a}.
\end{proof}
\vspace{-10pt}
\subsection{Proving Theorem \ref{thm.DerivativeFreeConvergence}.}

We first state and prove the following lemma:

\begin{lem} \label{lem.BoundednessOfFeasibleSet}
Suppose that the assumptions of Theorem \ref{thm.DerivativeFreeConvergence} hold, and let $C>0$ be any constant. Define $A_1 = \{\mathrm y\in\mathbb{R}^n:\ \|\E^T \mathrm y - \zeta^\star\| \le C\}$ and $A_2 = \{\mathrm y\in \mathbb{\R}^n:\ \sum_i k_i^{-1}(\mathrm y_i) = 0\}$. Then the set $A_1\cap A_2$ is bounded.
\end{lem}
\begin{proof}
First, we note that the inequality $\|\E^T\mathrm y - \zeta^\star\| \le C$ implies that for any edge $\{i,j\}\in\EE$, we have $|\mathrm y_i-\mathrm y_j| \le C+||\zeta^\star||$ by the triangle inequality. 
We let $\omega = (C+||\zeta^\star||)\mathrm{diam}(\G)$,  where $\mathrm{diam}(\G)$ is the diameter of the graph $\G$, so that if there exists some $i,j\in\V$ such that $|\mathrm y_i-\mathrm y_j| > \omega$ then $\mathrm y\not \in A_1$.
Moreover, let $\mathrm z = k(0)$, so $\sum_i k_i^{-1}(\mathrm z_i) = 0$, so that if $\mathrm y\in \mathbb{R}^n$ satisfies $\forall i:\ \mathrm y_i > \mathrm z_i$, then $\mathrm y\not \in A_2$. Indeed, for each $i$ we have $k_i^{-1}(\mathrm y_i) \ge k^{-1}_i(\mathrm z_i)$, and $k_{i_0}^{-1}(\mathrm z_{i_0}) > k_{i_0}^{-1}(\mathrm y_{i_0})$, meaning that $\sum_i k_i^{-1}(\mathrm y_i) > \sum_i k_i^{-1}(\mathrm z_i) = 0$. Similarly, if $\forall i, \mathrm z_i > \mathrm y_i$ then $\mathrm y\not\in A_2$.
We claim that for any $\mathrm y\in A_1\cap A_2$ and any $i\in \V$, we have $C_1 < \mathrm y_i < C_2$, where $C_1 = \min_j \mathrm z_j -\omega - 1$ and $C_2 = \max_j \mathrm z_j + \omega + 1$. Indeed, take any $\mathrm y\in \mathbb{R}^n$, and suppose that $\mathrm y_i \ge C_2$ for some $i\in \V$. There are two possibilities.
\begin{itemize}
\item There is some $k\in \V$ such that $\mathrm y_k < \max_j \mathrm z_j+1$. Then $|\mathrm y_i -\mathrm  y_k| > \omega$, implying that $\mathrm y\not \in A_1$.
\item For any $k\in \V$, $\mathrm y_k \ge \max_j \mathrm z_j +1$, implying that $\mathrm y \not \in A_2$.
\end{itemize}
Similarly, one shows that if there is some $i$ such that $\mathrm y_i \le C_1$, then $\mathrm y\not\in A_1\cap A_2$. This completes the proof.
\end{proof}

\begin{proof}[Proof of Theorem \ref{thm.DerivativeFreeConvergence}]
Consider the solution $\mathrm y (a)$ of $0 = k^{-1}(\mathrm y) + \E\diag(a)\gamma(\E^T \mathrm y)$ as a function of $a$. Then $\mathrm y(a)$ is a differentiable function by the inverse function theorem, and its differential is given by
$
\frac{d\mathrm y}{da} = -X(\mathrm y(a))\E\diag(\gamma(\E^T\mathrm y(a))),
$
where the matrix $X(\mathrm y)$ is given by 
\begin{align}
X(\mathrm y) = [\diag(\nabla k^{-1}(\mathrm y) )+ \E\diag(\nabla \gamma(\E^T\mathrm y))\E^T]^{-1}.
\end{align}
We note that $X(\mathrm y)$ is positive-definite for any $\mathrm y\in \mathbb{R}^n$, by Proposition~2 in \cite{Sharf2018b}. Thus, the gradient of $F$ is given by:\small
\begin{align}
\hspace{-6pt}\nabla F(a) = -\diag(\gamma(\E^T\mathrm y(a)))\E^TX(\mathrm y(a))\E(\E^T\mathrm y(a) - \zeta^\star).
\end{align}\normalsize
We note that $v^T\diag(\gamma(\E^T\mathrm y(a)) = \E^T\mathrm y(a) - \zeta^\star$, as $\gamma_e(x_e) = 0$ if and only if $x_e = \zeta^\star_e$ by strict monotonicity. Thus, \small
\begin{align}\label{eq.gradient_times_v}
v^T\nabla F(a) = -(\E(\E^T\mathrm y(a)-\zeta^\star))^TX(\mathrm y(a))\E(\E^T\mathrm y(a)-\zeta^\star)
\end{align}\normalsize
which is non-positive as $X(\mathrm y(a))$ is a positive-definite matrix.

Now, we claim that $v^T \nabla F(a) = 0$ if and only if $\E^T \mathrm y(a) = \zeta^\star$. Indeed, $\zeta^\star\in\IM(\E^T)$, so we denote $\zeta^\star = \E^T \mathrm y_0$ for some $\mathrm y_0$. As $X(\mathrm y(a))$ is positive definite, \eqref{eq.gradient_times_v} implies that $v^T\nabla F(a) = 0$ if and only if $\E(\E^T\mathrm y(a) - \zeta^\star) = \E\E^T(\mathrm y(a) - \mathrm y_0)$ is the zero vector. The kernel of the Laplacian $\E\E^T$ is the span of the all-one vector $\mathbb{1}_n$, so $\mathrm y(a) - \mathrm y_0 = \kappa\mathbb{1}_n$ for some $\kappa$, hence $\E^T \mathrm y(a) = \E^T \mathrm y_0 = \zeta^\star$. This concludes the first part of the proof.

As for convergence, we know that if $h$ is small enough, then $F(a^{(j+1)}) < F(a^{(j)})$. However, the value of $h$ so that $F(a^{(j+1)}) < F(a^{(j)})$ can depend on $a^{(j)}$ itself, but it is obvious that if $h$ is small enough, then for any $j$, we have $F(a^{(j)}) \le F(a^{(0)})$. We let $C = F(a^{(0)})=\|\E^T\mathrm y(a^{(0)}) - \zeta^\star\|$, and consider the sets $A_1 = \{\mathrm y:\ \|\E^T\mathrm y - \zeta^\star\| \le C\}$ and $A_2 = \{\mathrm y:\ \sum_i k_i^{-1}(\mathrm y_i) = 0\}$. 
For any $j$, we know that $\mathrm y(a^{(j)}) \in A_1$ by above, and that $\mathrm y(a^{(j)}) \in A_2$ by the steady-state equation $0 = k^{-1}(\mathrm y(a)) + \E\diag(a)\gamma(\E^T \mathrm y(a))$. Thus, all steady-state outputs achieved during the algorithm are in the set $\mathcal D = A_1\cap A_2$, which is bounded by Lemma \ref{lem.BoundednessOfFeasibleSet}. The map sending a matrix to its minimal singular value is continuous, meaning that $\underline{\sigma}(X(\mathrm y))$ achieves a minimum on the set $\mathcal D$ at some point $\mathrm y_1$, and the minimum is positive as $X(\mathrm y_1)$ is positive-definite. We denote the minimum value by $\underline{\sigma}(\mathcal D)$.

 Now, consider equation \eqref{eq.gradient_times_v}. We get that $v^T\nabla F(a)$ is bounded by above $-\underline{\sigma}(\mathcal D)||\E(\E^T\mathrm y(a) - \zeta^\star)||^2$. In turn, we saw above that unless $\E^T \mathrm y(a) = \zeta^\star$, $\E(\E^T\mathrm y(a) - \zeta^\star) \neq 0$, meaning that $\|\E(\E^T\mathrm y(a) - \zeta^\star)\| \ge \varsigma||\E^T\mathrm y(a)-\zeta^\star||^2$, where $\varsigma$ is the minimal nonzero singular value of $\E$. Hence, at any time step $j$, $v^T \nabla F(a^{(j)}) < -\underline{\sigma}(\mathcal D)\varsigma F(a^{(j)})$. In turn we conclude that $F(a^{(j+1)}) = F(a^{(j)}) - h\underline{\sigma}(\mathcal D)\varsigma F(a^{(j)}) + \mathcal{O}(h) = (1-h\underline{\sigma}(\mathcal D)\varsigma)F(a^{(j)}) + \mathcal{O} (h)$. Iterating this equation shows that eventually, $F(a^{(j)}) < \varepsilon$, completing the proof.
\end{proof}
\if(0)
\begin{IEEEbiography}[{\includegraphics[width=1in,height=1.25in,clip,keepaspectratio]{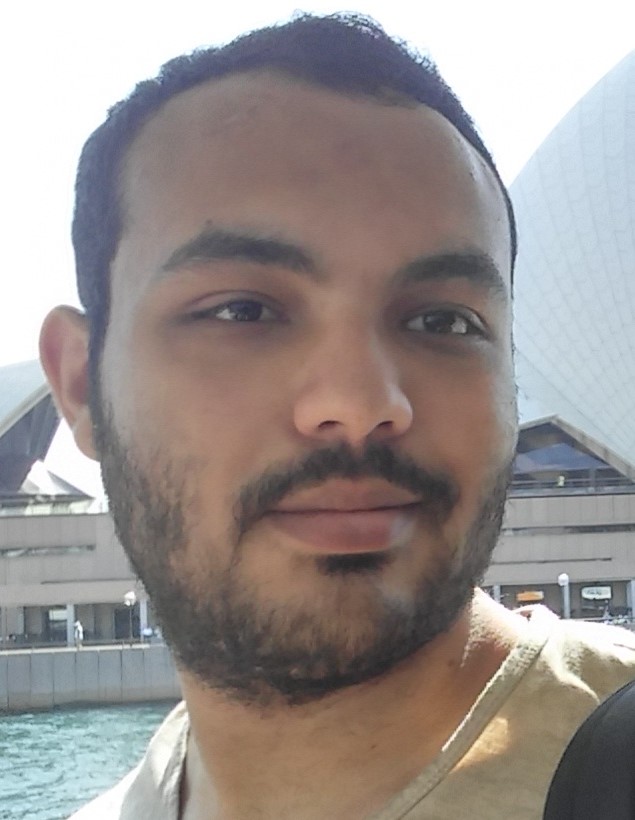}}]
 {\bf Miel~Sharf} received his Ph.D. from the Aerospace Engineering department at the Technion - Israel Institute of Technology in 2020. He received his B.Sc. (13`) and M.Sc. (16`) degrees in Mathematics from the Technion - Israel Institute of Technology. He is a recipient of the Springer Thesis Award.

Since September 2020, he is a postdoctoral  researcher with the Division of Decision and Control Systems, KTH Royal Institute of Technology, Stockholm, Sweden. His research interests include the relation between graph theory and algebraic graph theory to multi-agent systems, nonlinear control and passivity theory, data-driven control, and security in networked systems.
 \end{IEEEbiography}
 \vspace{-5pt}
\begin{IEEEbiography}[{\includegraphics[width=1in,height=1.25in,clip,keepaspectratio]{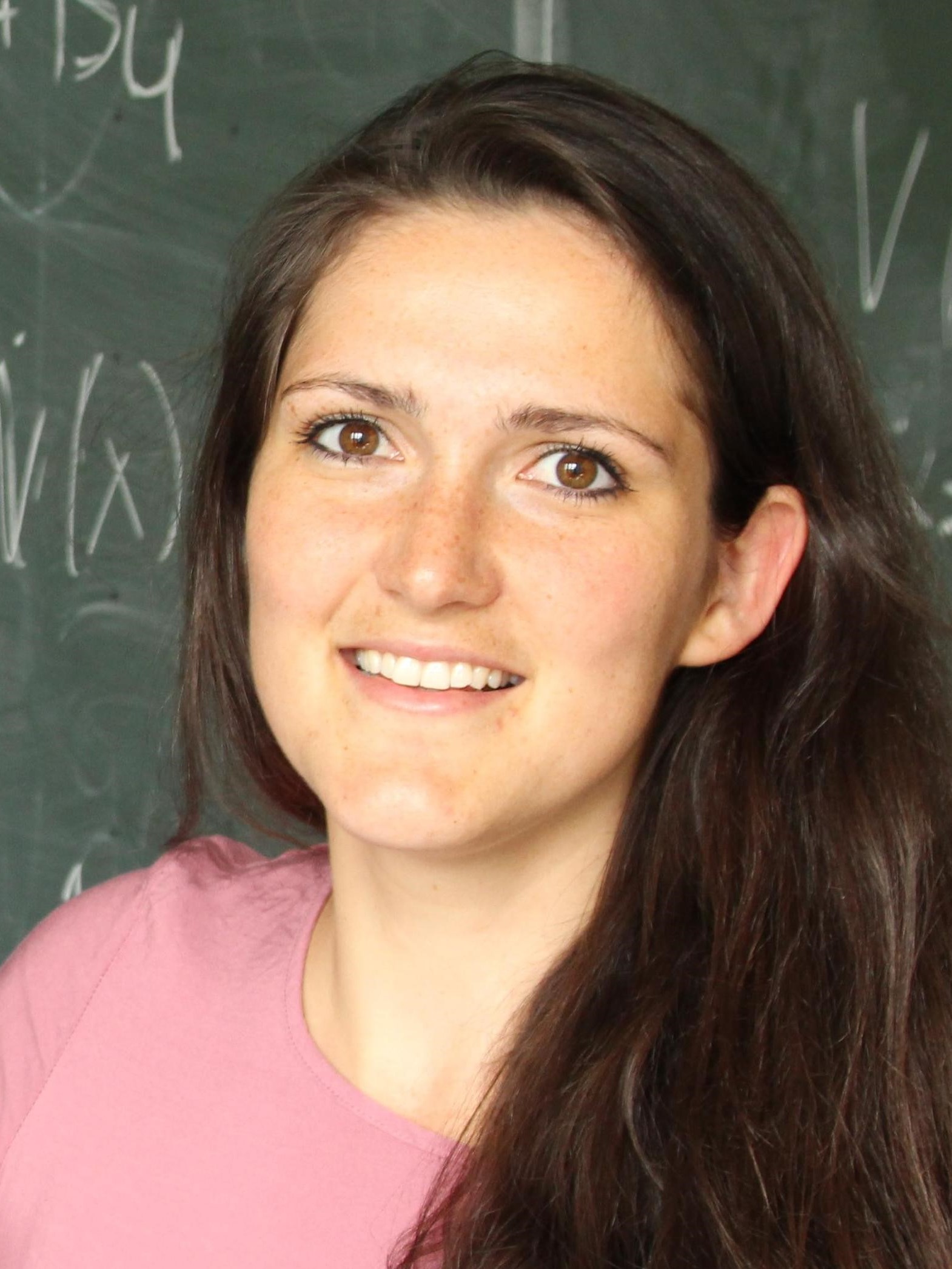}}]{Anne Koch (n\'ee Romer)}
received the M.Sc. in Engineering Science and Mechanics from the Georgia Institute of Technology, Atlanta, USA, in 2014, and the M.Sc. in Engineering Cybernetics from the University of Stuttgart, Germany, in 2016. \\
In 2016, she joint the
Institute for Systems Theory and Automatic Control at the University of Stuttgart as a research and teaching assistant pursuing the Ph.D. degree. 
Her re\-search interests include data-based systems analysis and controller design.
\end{IEEEbiography}
\vspace{-5pt}
 \begin{IEEEbiography}[{\includegraphics[width=1in,height=1.25in,clip,keepaspectratio]{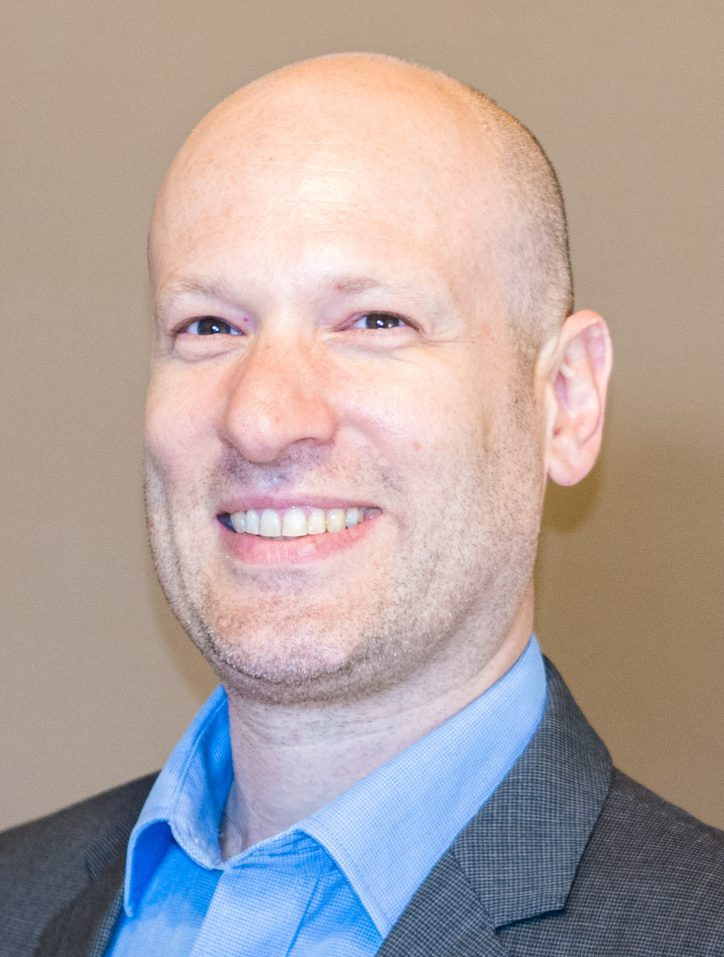}}]
 {\bf Daniel~Zelazo} is an associate professor of Aerospace Engineering at the Technion - Israel Institute of Technology. He received his BSc. (99) and M.Eng (01) degrees in Electrical Engineering from the Massachusetts Institute of Technology. In 2009, he completed his Ph.D. from the University of Washington in Aeronautics and Astronautics. From 2010-2012 he served as a post-doctoral research associate and lecturer at the Institute for Systems Theory \& Automatic Control in the University of Stuttgart. His research interests include topics related to multi-agent systems, optimization, and graph theory.
 \end{IEEEbiography}
 \vspace{-5pt}
\begin{IEEEbiography}[{\includegraphics[width=1in,height=1.25in,clip,keepaspectratio]{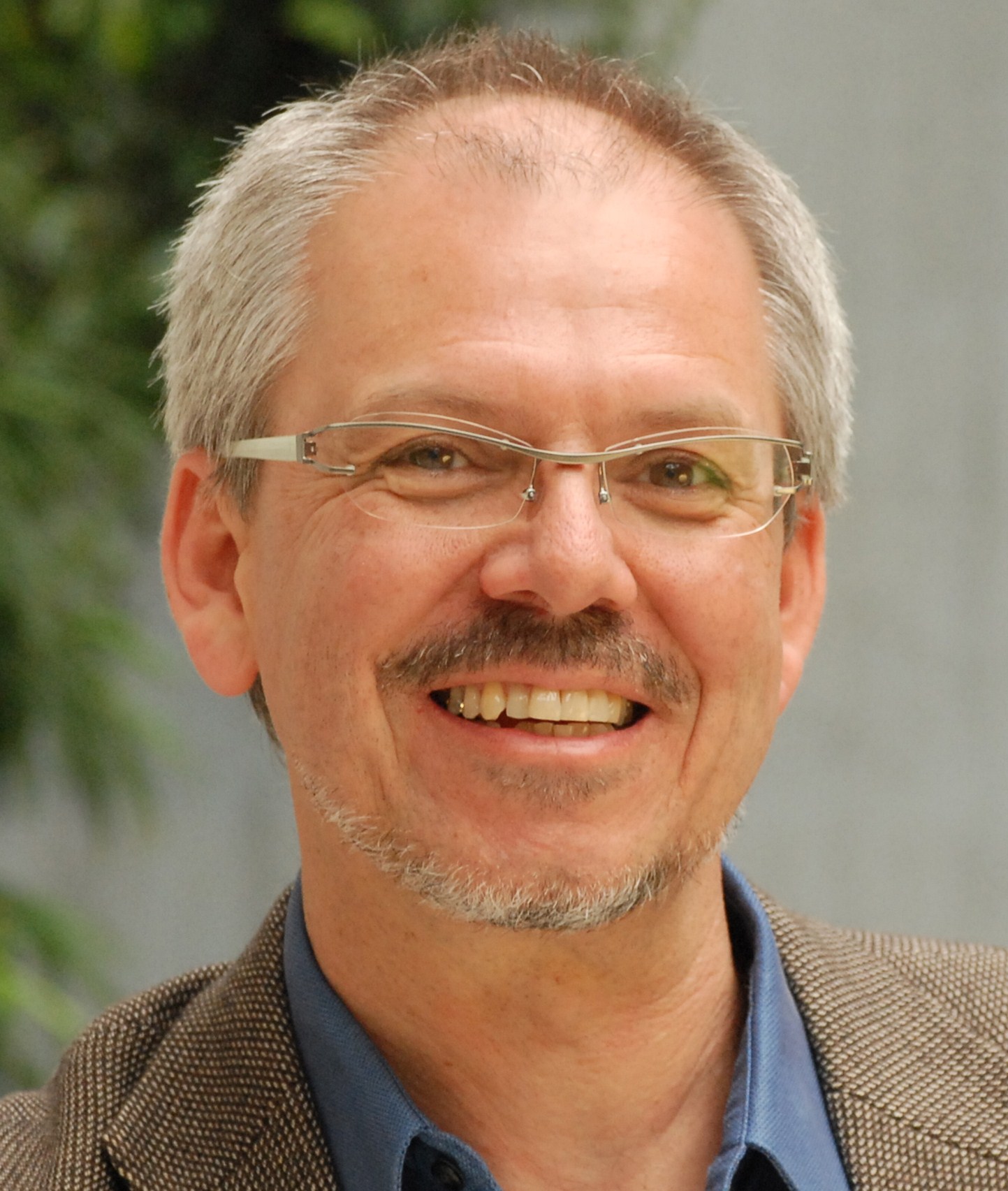}}]{Frank Allg\"{o}wer}
is professor of mechanical engineering at the University of Stuttgart, Germany, and Director of the Institute for Systems Theory and Automatic Control (IST) there.\\ 
Frank is active in serving the community in several roles: Among others he has been President of the International Federation of Automatic Control (IFAC) for the years 2017-2020, Vice-president for Technical Activities of the IEEE Control Systems Society for 2013/14, and Editor of the journal Automatica from 2001 until 2015. From 2012 until 2020 Frank served in addition as Vice-president for the German Research Foundation (DFG), which is Germany’s most important research funding organization. \\
His research interests include predictive control, data-based control, networked control, cooperative control, and nonlinear control with application to a wide range of fields including systems biology.
\end{IEEEbiography}
\fi

\end{document}